\documentclass[10 pt, journal]{IEEEtran}  

\IEEEoverridecommandlockouts     

\usepackage{amsmath,amsthm,amstext,amsfonts,amssymb,mathrsfs}
\usepackage{graphicx,caption,subcaption,algorithm,algorithmic}

\usepackage{tikz, xcolor}
\usepackage{tkz-euclide}
\usetikzlibrary{shapes,arrows,calc,intersections,through,backgrounds}

\usepackage{pgfplots}
\pgfplotsset{compat=newest}
\usetikzlibrary{plotmarks}
\usetikzlibrary{arrows.meta}
\usepgfplotslibrary{patchplots}
\usepackage{grffile}
\usepackage{amsmath}

\usepackage{setspace}
\usepackage{array} 
\usepackage{cases}
\usepackage{bbm} 
\usepackage[sort,compress]{cite}

\newtheorem{lemma}{Lemma}
\newtheorem{assumption}{Assumption}

\newtheorem{proposition}{Proposition}
\newtheorem{theorem}{Theorem}
\newtheorem{definition}{Definition}

\newcommand{\tn}{\textnormal}

\title{Selling Demand Response Using Options} 

\author{Deepan Muthirayan$^d$, Dileep Kalathil$^c$, Sen Li$^b$, Kameshwar Poolla$^a$, Pravin Varaiya$^a$  
\thanks{$^a$ Electrical Engineering and Computer Sciences, University of California, Berkeley, CA 94720. Email: \{poolla, varaiya\}@berkeley.edu}
\thanks{$^b$ Department of Civil and Environmental Engineering, The Hong Kong University of Science and Technology, Clear Water Bay. Email: cesli@ust.hk}
\thanks{$^c$ Electrical Engineering and Computer Sciences, University of Texas A and M, TX 78363. Email: dileep.kalathil@tamu.edu}
\thanks{$^d$ Electrical Engineering and Computer Sciences, University of California, Irvine, CA 92697. Email: dmuthira@uci.edu}
\thanks{Supported in part by the National Science Foundation under Grants EECS-1129061/9001, 
CPS-1239178.}
}

\begin{document}

\maketitle

\markboth{Submission for IEEE Transactions on Smart Grid} 
{Kalathil \emph{et al.}} 

\begin{abstract}
Wholesale electricity markets in many jurisdictions use a two-settlement structure:  a day-ahead market for bulk power transactions and a real-time market for fine-grain supply-demand balancing. This paper explores trading demand response assets within this two-settlement market structure. We consider two approaches for trading demand response assets: (a) an intermediate spot market with contingent pricing, and (b) an over-the-counter options contract. In the first case, we characterize the competitive equilibrium of the spot market, and  show that it is socially optimal. Economic orthodoxy advocates spot markets, but these require expensive infrastructure and regulatory blessing.In the second case, we characterize competitive equilibria and compare its efficiency with the idealized spot market. Options contract are private bilateral over-the-counter transactions and do not require regulatory approval. We show that the optimal social welfare is, in general, not supported. We then design optimal option prices that minimize the social welfare gap. This optimal design serves to approximate the ideal spot market for demand response using options with modest loss of efficiency. Our results are validated through numerical simulations. 
\end{abstract}

\section{Introduction}

Wholesale electricity markets in many jurisdictions use a two-settlement structure: a day-ahead market for bulk power
transactions and a real-time market for fine-grain supply-demand balancing.  Forecast errors in the day-ahead market necessitate subsequent balancing in the real-time market. With deeper penetrations of wind and solar generation, markets must be able to contend with greater levels of uncertainty stemming from renewable intermittency. Forecast errors increase, and balancing supply and demand becomes more challenging. The traditional approach of balancing using conventional fossil fuel based reserves is untenable: it is expensive and defeats the emissions benefits of renewables. Balancing the variability of intermittent renewable generation through demand flexibility is a far better alternative to reserve generation, as it produces no emissions and consumes no resources. This is recognized and encouraged by the Federal Energy Regulatory Commission (FERC) through its Order 745, which mandates that demand response be compensated on par with the conventional generation that supplies grid power \cite{federal2012assessment}. Commercial buildings, light industry, and households are flexible in their electricity consumption. These agents can be induced to yield this flexibility  in exchange for monetary compensation. This paper explores trading demand response assets within a traditional  two-settlement market structure. 

We consider the setting where a {\it Load Serving Entity} (LSE) supplies electricity to a collection of consumers at the delivery time $T$.  An aggregator manages the aggregate load flexibility of these consumers. The LSE interacts directly with the aggregator and can request a certain aggregate load reduction which will be reliably produced at the delivery time. The LSE can purchase bulk power in the day-ahead market and can also buy balancing 
power in the real-time market. It also has access to zero marginal cost renewable generation. 
We consider the situation where excess renewable generation is spilled, and cannot be sold back into the real-time market. 
Other generalizations of our results are possible, but we choose to explore the simplest situation.

{\em When should demand response assets be traded?}
Well in advance of the delivery time, the LSE has poor forecasts of its renewable generation and of clearing
prices in the real-time market. So the LSE prefers to delay its demand response request close to the delivery time.
Conversely, the aggregator prefers to receive any load curtailment requests well before the delivery time. This
affords its client consumers sufficient lead time to organize their electricity use and cede their demand reduction.
These considerations argue that demand response assets should be traded in an intermediate market as a 
recourse between the day-ahead and real-time markets. 

{\em What is an appropriate mechanism for the intermediate time trading of demand response assets?}
Economic orthodoxy argues in favor of an idealized spot market with contingent prices from the perspective of efficiency. 
In this intermediate spot market,  trading takes place after counter-parties digest all information that is revealed. 
Therefore, the clearing prices are {\em contingent} on the realized information. 
While the spot market is efficient, it has two main drawbacks:  (a) pricing is typically very volatile and does not 
offer guaranteed income to demand response assets to compensate for yielding their load flexibility and for the associated capital costs, 
and (b) an intermediate spot market requires organized infrastructure and regulatory approval which can be very expensive. 

To overcome these difficulties, we propose to trade demand response assets using call options.  
In our scheme, the LSE buys a number of call options contracts from the aggregator at time $t_o$, 
coincident with the gate closure of the day-ahead market.
It pays an option price $\pi^o$ per contract. Each call option contract affords the LSE the right, but not the obligation, to 
receive one unit of  load reduction from the aggregator. These options expire at the intermediate time $t_1$
by which time they must be exercised or forfeited. To exercise these options the LSE must pay the aggregator 
the strike price $\pi^{sp}$ per unit of load reduction. The strike price is not contingent, it is fixed and known at time $t_o$.
Payment from the sale of option contracts provides a guaranteed income to flexible loads 
for their demand response {\em capability}. Subsequent payment from the exercise of option contracts compensates loads for the {\em provision} of demand response. Since option contracts can be viewed as private over-the-counter transactions 
between the LSE and the aggregator, our scheme does not require regulatory blessing  or organized market infrastructure. 

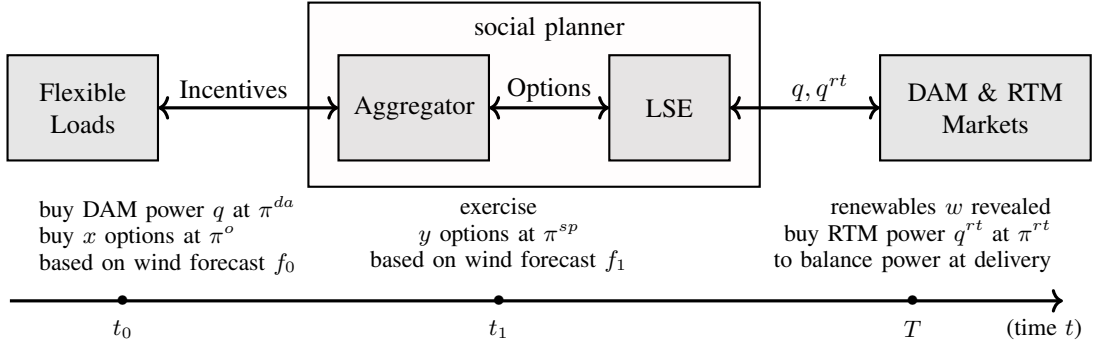
\begin{figure*}[htb] 
\label{fig:timeline}
\centering
\begin{tikzpicture}[xscale=0.4, yscale=0.35]
\draw [black, fill=black, fill opacity = 0.1, thick] (-9,4) rectangle (-4,8);
\node [align=center] at (-6.5,6) {Flexible \\ Loads};
\draw [<->,very thick] (-4,6) --  (2,6);
\node [align=center] at (-1.5,6.7) {Incentives}; 
\draw [black, fill=black, fill opacity = 0.1, thick] (2,4) rectangle (7,8);
\node [align=center] at (4.5,6) {Aggregator};
\draw [<->,very thick] (7,6) --  (11,6);
\draw [black, fill=black, fill opacity = 0.1, thick] (11,4) rectangle (15,8);
\node [align=center] at (13,6) {LSE};
\draw [<->,very thick] (15,6) --  (20,6);
\draw [black, fill=black, fill opacity = 0.1, thick] (20,4) rectangle (27,8);
\node [align=center] at (23.5,6) {DAM \& RTM \\ Markets};
\draw [black, fill=red, fill opacity = 0.01, thick] (1,3) rectangle (16,10);
\node [align=center] at (9,9) {social planner}; 
\node [align=center] at (9,6.7) {Options}; 
\node [align=center] at (18,6.75) {$q,  q^{rt}$}; 
\end{tikzpicture}
\begin{tikzpicture}[scale=0.5]
{\small
\draw [->,very thick,black] (0,-0.25) --  (28,-0.25);
\node [align=center] at (27.5,-1) {(time $t$)};
\node [align=left] at (4.2,1.5) {buy DAM power $q$ at $\pi^{da}$ \\ buy $x$ options at $\pi^o$ \\ based on wind forecast $f_0$}; 
\node [align=center] at (3,-1) {$t_0$}; 
\node [align=center] at (13,1.5) {exercise \\ $y$ options at $\pi^{sp}$ \\ based on wind forecast $f_1$ }; 
\node [align=center] at (13,-1) {$t_1$}; 
\node [align=right] at (24,1.5) {renewables $w$ revealed \\ buy RTM power $q^{rt}$ at $\pi^{rt}$ \\ to balance power at delivery}; 
\node [align=center] at (24,-1) {$T$}; 
\node [align=center] at (3,-0.25) {$\bullet$};
\node [align=center] at (13,-0.25) {$\bullet$};
\node [align=center] at (24,-0.25) {$\bullet$};}
\end{tikzpicture}
\caption{Players, interactions, and decision time-line.} \label{fig:timeline}
\end{figure*}

\subsection{Our Contributions} 
Our principal contributions are:
\begin{itemize}
\item First, we consider  optimal energy scheduling from the perspective of a social planner. 
We formulate this as a three stage optimization problem and characterize the optimal decisions at each stage: the optimal energy purchase in the day-ahead market, the optimal demand response (or load curtailment) decision at the intermediate stage, 
and the balancing energy purchase in the real-time market. This serves as a benchmark for evaluating other market designs.
\item Second, we consider an intermediate spot market with contingent pricing. We study the interactions of the LSE and the aggregator in a spot-market. We show that there exists a competitive equilibrium, and the equilibrium is socially optimal, i.e., it realizes the same system cost as the benchmark. 
\item Third, we study the options market for the LSE and the aggregator. We show that under some conditions, a competitive equilibrium always exists, and it is the optimal solution to a convex optimization problem. We compare the efficiency of the equilibrium for the options market and the spot market, and show that the options market is not necessarily socially optimal. 
We then design optimal option and prices which minimize the welfare gap at the competitive equilibrium. 
\end{itemize}
\subsection{Related Work} 

There is extensive literature on demand response and work related to managing the uncertainty with renewable integration \cite{li2011optimal, roscoe2010supporting, samadi2014real, chen2012real, wu2012vehicle, mohsenian2010autonomous, yang2013game, gabriel2006optimal, haring2013decentralized, fahrioglu2000designing, gatsis2012residential, chavali2014distributed, shi2014optimal, lidemand, xudemand, huang2012optimal}. These works can be broadly classified as price-based or contract-based.
  
{\it Price-based Demand Response}: This is a type of demand response where the consumers alter their energy consumption based on time varying prices determined apriori by the LSE. The objective here is to improve overall system benefits by influencing the consumers to shift their demand. The works in \cite{li2011optimal, roscoe2010supporting, samadi2014real} propose different approaches to determine the time varying prices such that the overall system benefits, measured in terms of efficiency and load variability, are improved. The authors in \cite{mohsenian2010autonomous, yang2013game} study a game theoretic formulation and propose a pricing strategy that improves system benefits in Nash equilibrium. Closely related works such as \cite{wu2012vehicle} propose a time varying price policy to utilize flexible storage of EVs in order to manage load variability. Other works such as \cite{chen2012real} propose a demand response management strategy using a stochastic optimization procedure that accounts for financial risks associated with time varying prices.

{\it Distributed Price-based Demand Response}: Authors in \cite{gatsis2012residential}, \cite{chavali2014distributed} and \cite{shi2014optimal} propose iterative distributed load control schemes with the objective of meeting system requirements and minimizing consumer discomfort. They primarily address the coordination of multiple demand response users by iteratively discovering the most appropriate electricity price and its variation with time. 

The setting we consider is different from the above works, which are primarily concerned with price-responsive demand response. We propose a market mechanism for direct calling of a certail level of demand response instead of using price to influence demand. 
Here, a LSE can buy DR contracts from aggregators of DR in the day-ahead market and can determine the amount of DR to call in the real-time market at an intermediate stage when more information is available on renewable generation at real-time. 

{\it Multi-stage Stochastic Decision}: Varaiya \emph{et al.} \cite{Va-RLD11} propose a {\em risk-limiting dispatch}  
approach for integrating renewable energy in the grid. They formulate a multi-stage stochastic control problem where at each stage the utility makes purchase decisions based on the available information. Rajagopal \emph{et al.} \cite{rajagopal2013risk} extend this approach and characterize optimal power procurement policies as threshold based decisions. 
Our work parallels the approach of Varaiya et al. \cite{Va-RLD11}. In particular, we extend their approach to a contract setting as proposed in this paper, where the decision of two entities are coordinated in a multi-stage decision problem through an options contract mechanism.

{\it Contract-based Demand Response}: The works in \cite{muthirayan2019mechanism, haring2013decentralized, fahrioglu2000designing} address the problem of demand response aggregation from a mechanism design perspective. The objective of the mechanism design is to gather demand response contracts at minimal cost and at preferrably maximal privacy so that the aggregator or the LSE can meet the DR requirements of the system. Alternatively, demand response contracts that treat demand response as a differentiated good, based on their power level and duration, have also been proposed \cite{nayyar2016duration}, \cite{bitar2017deadline}. Our work is different from these set of works in the sense that we provide a multi-settlement market framework for trading the aggregated demand response in the electricity markets. 

{\it Options}: Authors in \cite{eydeland1999fundamentals, kluge2006pricing, deng2006electricity, haarbrucker2009valuation, pflug2009electricity} discuss the pricing of electricity options for hedging price risks in deregulated electricity markets. Authors in \cite{oren2001integrating} provide a forward electricity contract with a call option to hedge against price risks while exercising load flexibility. Authors in \cite{kamat2002exotic} provide a similar forward contract bundled with a option to hedge against price for interruptible services. In \cite{chung2003electricity} a forward contract with a bilateral financial option for buyers and sellers to take advantage of their flexibility and reduse their price risks is discussed. In \cite{oum2006hedging}, a options based hedging mechanism for a LSE is discussed and in \cite{pineda2012managing} an options contract for a producer is discussed. There are works such as \cite{oren2005generation} that discuss long term options contract as a mechanism for ensuring generation adequacy. 
Other works such as \cite{hsu1998spark, deng2001exotic, sezgen2007option} provide option models for assessment of the value of generation and demand response investment decisions. Authors in \cite{cesena2013real} provide a review of application of real option models for valuing electricity generation projects and renewable energy projects. In this work we propose options market for trading aggregated DR in an intermediate market that allows the LSE to call for load curtailment based on an improved forecast of the wind power that is available at this intermediate time. The options market is proposed as an alternative to the intermediate spot market that can be cumbersome to organize and volatile in terms of the revenue it generates for the service providers.

The paper is organized as follows. We introduce the basic notions and notation in Section \ref{sec:pelims}. In Section \ref{sec:opt-sch}, we consider the problem of energy scheduling with demand response from the perspective of a system planner. In Section \ref{sec:CE-contingent} we discuss the implementation of intermediate spot market for scheduling demand response. We present the options market mechanism in Section \ref{sec:CE-option}. Finally, we conclude the paper with a brief description of future research directions in Section \ref{sec:conclusion}.   

\section{Preliminaries}  
\label{sec:pelims}
The setting we consider is shown in Figure \ref{fig:timeline}. A load serving entity (LSE)
supplies $l$ units of electricity to a collection of consumers for delivery at time $T$. The demand $l$ is considered inelastic and known at time $t_o$. Indeed, day-ahead load forecast errors are within $1\%-2\%$\cite{CAISO2018briefing}.  
The LSE buys $q$ units of energy in the day-ahead market at price $\pi^{da}$. At the intermediate time $t_1$, he extracts $y$ units of demand response, which incurs a disutility of $\phi(y)$. The LSE has access to zero marginal cost random renewable 
generation $w$ which is realized at the delivery time $T$. 
To meet its demand obligations, the LSE buys the remaining energy required $q^{rt}$ in the real-time market at price $\pi^{rt}$.  
The total energy purchase must satisfy:
\begin{equation}
\label{eq:power-balance}
l \leq q + q^{rt} + w + y.
\end{equation}
Note that we consider the situation where excess renewable generation is spilled, 
and cannot be sold back into the real-time market. This is necessary to ensure that the LSE does not sell all of the renewable generation back in the real-time market and can be imposed as a regulation. The demand response purchase made at the intermediate time $t_1$ is based on a forecast $f_1$ of $w$ and forecast of $\pi^{rt}$.  

\subsection{Model Uncertainties}
Let $f_{0}$ denote the information available at $t_o$. Let $p(w|S)$ be the conditional probability of the wind given the intermediate forecast state $S$ at time $t_{1}$.  The forecast state $S$ can be regarded as a sufficient statistic which parameterizes the information on wind at time $t_{1}$. We parameterize $S \in [0, 1]$. We call this an \emph{information state}. Define 
\[p_{s}(w) = p(w | S=s),~~P_{s}(z) = \int^{z}_{w=0} p_{s}(w) dw, \]
where $P_{s}(z)$ is the  probability that the wind at time $T$ is less than $z$ given the information state $s$. Let $\alpha(s)$ be the prior probability density function of the information state, i.e.,  
\[\alpha(s) = \mathbb{P}\left(S = s |  f_{0}\right)\]
We assume that real-time price $\pi^{rt}$ is a random variable and denote the expected real-time price conditioned on the information state by,
\[\overline{\pi}^{rt}_{s} = \mathbb{E}[\pi^{rt}|S = s] \] 
The day-ahead price $\pi^{da}$ is  known at time $t_{0}$. We use $\mathbb{E}_{S}\left[\cdot\right]$ and $\mathbb{E}_{w}\left[\cdot\right]$ to denote the expectation with respect to the information state and the randomness in wind, respectively. Let $\mathbb{E}\left[\cdot\right]$ denote the joint expectation. 
We make the following assumptions.
\begin{assumption}
\label{as:w-pirt}
(i) $\mathbb{P}(w \geq z | S = s)   < \mathbb{P}(w \geq z | S=s'),$   $ \forall z,~ \text{if}~ s' > s$, 
(ii) $\overline{\pi}^{rt}_{s'} < \overline{\pi}^{rt}_{s}$, if $s'  > s$, 
(iii) $\pi^{rt}$ and $w$ are conditionally independent given the information state $s$.
\end{assumption}
Assumption (i) imposes a stochastic ordering on wind conditioned on the information state $s \in [0, 1]$. The intuitive interpretation is that larger values of $s$ indicate (stochastically) more wind. This assumption guarantees that $P_{s'}(z) < P_{s}(z), \forall z$, if $s' > s$ so that the cumulative distribution $P_{s}(\cdot)$ and $P_{s'}(\cdot)$ do not intersect. Assumption (ii) similarly imposes an ordering on the expected value of the real-time price conditioned on the information state. The ordering is such that higher values of $s$ correspond to a lower expected real-time price (because more wind power reduces demand in the real-time market).  Assumption (iii) imposes that the information state $s$ contains all the causal factors that determine the real-time price and wind power $w$. This is a reasonable assumption because the information state $s$ represents the underlying state of nature. 
%

\subsection{Decision Making of Players}
The players of the problem include an LSE, an aggregator, and a social planner. We model them as follows.

\noindent {\bf Load Serving Entity (LSE):} The LSE is responsible for satisfying the energy balance specified by equation \eqref{eq:power-balance}. At time $t_0$, it buys $q$ units of energy at a price $\pi^{da}$ from the day-ahead market. At time $t_1$, it receives a load reduction of $y_{s}$ units from the aggregator when the information state $s$ is revealed, and makes a  payment $R_{s}(y_{s})$. At time $T$, the renewable $w$ is revealed, and the LSE purchases the remaining energy from the real-time market, i.e.,  $q^{rt} = (l-q-y_{s}-w)_{+}$. The ex-post cost for the LSE given the information state $s$ is,  
\begin{equation} 
\label{eq:J-LSE-expost}
{J}^{LSE}_{s} = \pi^{da} q + R_{s}(y_{s}) + \pi^{rt} (l-q-y_{s}-w)_{+}
\end{equation}


\noindent {\bf Aggregator:} The aggregator suffers a disutility $\phi(y_s)$  for a load reduction of $y_s$ units,  and receives a compensation payment $R_s(y_s)$ from the LSE. The ex-post cost for the aggregator, given the information state $s$, is as follows,
\begin{equation}
\label{eq:J-agg-expost}
{J}^{agg}_s = \phi(y_s)  - R_{s}(y_s)
\end{equation}
We assume that the disutility function satisfies the assumption given below.
\begin{assumption}
\label{disutilityfunction}
$\phi(y_s)$ is  twice differentiable, and is strongly convex in $y_s$, i.e.,  $\phi''(y_s)> 0$.
\end{assumption}

\noindent {\bf Social Planner or Entity (e):} We consider a hypothetical agent, the \emph{social planner}, which combines the roles of the LSE and the aggregator. We denote decision variables and cost functions  of the social planner with the superscript e for entity. This social planner buys $q$ units of energy from the day-ahead market, receives a load curtailment of $y_s$ units at an intermediate time $t_{1}$,  acquires zero marginal-cost realized wind power $w$ at time $T$, and purchases the remaining energy $(l-q-y_s-w)_{+}$ from the real-time market for load balance. Given $s$, the ex-post cost for the social planner (also called the system cost) is:
\begin{equation}
J^{e}_s = \pi^{da} q + \phi(y_s) + \pi^{rt} (l-q-y_s-w)_{+}
\end{equation}
Payment for demand response is an internal exchange between the LSE and aggregator, and does not appear in the social planner ledger.  In the sequel, we first discuss the optimal scheduling problem for the social planner, and then we study the interaction between the LSE and the aggregator in the intermediate market and options market. We characterize the competitive equilibrium in both markets, and compare the system costs.   

\section{Optimal Scheduling for the Social Planner}  
\label{sec:opt-sch} 
This section studies the optimal scheduling of energy from the perspective of the social planner. We separately consider the scheduling problems with and without demand response. We use these solutions as benchmarks to compare the various market mechanisms we propose in subsequent sections.  

\subsection{Optimal Scheduling without Demand Response}
\label{sec:opt-sch-ndr}
In the absence of demand response, the social planner is confined to purchase energy from the day-ahead and real-time markets. Let $J^{e}_{ndr}(q)$ be the expected cost for the social planner in the absence of demand response. This is a function of the day-ahead purchase $q$ and is 
\begin{equation}
\label{eq:social planner-J-ndr-1}
J^{e}_{ndr}(q) =  \pi^{da} q + \mathbb{E}[{\pi}^{rt} (l-q-w)_{+}]  
\end{equation}
This implicitly accounts for the balance inequality (\ref{eq:power-balance}) necessary to service the load $l$. The optimal decision of the social planner is
\begin{equation}
\label{eq:optimaldecision4SP}
q^{e}_{ndr} = \arg \min_{q\geq 0} J^{e}_{ndr}(q)
\end{equation}
We have the following:
\begin{proposition}
\label{thm:no-dr-social planner}
$J^{e}_{ndr}(\cdot)$ is  convex. The minimizer $q^e_{ndr}$ solves
\begin{equation}
\label{optimalqwithoutDR}
\pi^{da} -  \mathbb{E}_{s}[\overline{\pi}^{rt}_{s} P_{s}(l-q^{e}_{ndr})]  = 0
\end{equation}
\end{proposition}

\begin{assumption}
To avoid trivial results, we assume that the day-ahead market price is discounted from the expected real-time market price, i.e., $\pi^{da} <  \mathbb{E}_{s}[\overline{\pi}^{rt}_{s} P_{s}(l) ]$. This will ensure that $q^{e}_{ndr} > 0$. 
\end{assumption} 

\subsection{Optimal Scheduling with Demand Response}  
\label{sec:opt-sch-dr}

With demand response, the net expected cost for the social planner as a function of the first-stage purchase is given by,  
\begin{equation}
\label{eq:social planner-J-dr}
J^{e}(q) = \pi^{da} q + \mathbb{E}_{s}\left[\min_{y_s\geq 0}  J^{e}_s(y_s; q)  \right]
\end{equation} 
where $J^{e}_s(y_s; q) $ is the expected second-stage cost conditioned on $s$ and  $q$ and is given by,
\begin{equation}
\label{eq:social planner-J2}
J^{e}_{s}(y_s;q) = \phi(y_s) + \mathbb{E}_{w} \left[ \overline{\pi}^{rt}_{s} (l-q-y_s-w)_{+} \vert s\right]
\end{equation}
The optimal first-stage and second-stage decisions, $q^{e}$ and $y^{e}_s$ respectively, are 
\begin{align}
\label{eq:social planner-q-y-1}
\begin{cases}
q^{e} = \arg \min_{q\geq 0} J^{e}(q), \\ y^{e}_s = \arg \min_{ 0\leq y_s\leq l} J^{e}_s(y_s;q^e)
\end{cases} 
\end{align}
The optimal system cost is then $J^{*e} = J^{e}(q^{e})$. Using the fact that both (\ref{eq:social planner-J-dr}) and (\ref{eq:social planner-J2}) are convex, we can solve for $q^e$ and $y_s^e$ using the conditions given in the following proposition.
\begin{proposition} 
\label{thm:dr-social planner} 
${J}^{e}(\cdot)$ and ${J}^{e}_s(\cdot)$ are convex. For any given first stage decision $q$, the second-stage decision $y^{e}_s$ is given by,
\begin{align}
\label{eq:social planner-ys-1}
\begin{cases}
\phi'(y_s^e) = \overline{\pi}^{rt}_{s} P_{s}(l-q-y_s^e),  & \text{if } \phi'(y_s^e) < \overline{\pi}^{rt}_{s} P_{s}(l-q) \\
y_s^e=0,  & \text{if } \phi'(y_s^e) \geq  \overline{\pi}^{rt}_{s} P_{s}(l-q)
\end{cases}
\end{align}
The first-stage decision $q^{e}$ is given by the solution of,
\begin{equation}
\label{eq:social planner-qe-1}
\pi^{da} - \mathbb{E}_s[ \overline{\pi}^{rt}_{s} P_{s}(l-q^{e}-y^{e}_s)]  = 0
\end{equation}
\end{proposition}  
The proof is offered in Appendix B.

\subsection{Socially Optimal Costs}
The optimal costs for the social planner with and  without demand response are
\[ J^*_{dr}=J^e(q^e) 
\text{ and } J_{ndr}^*=J^e_{ndr}(q^e_{ndr})\]
respectively. Clearly, $J^*_{dr}\leq J^*_{ndr}$. These social cost values serve as benchmarks for our market mechanism designs.  In Section \ref{sec:CE-contingent}, we show that a spot market with contingent prices realizes the socially optimal cost $J^*_{dr}$. In Section V, we show that trading demand response in an options market will, in general,  result in a loss of social welfare. We further select options prices so that this welfare gap is modest. As a result, the over-the-counter options market can well approximate the idealized spot market. 

\section{Spot Markets with Contingent Prices}      
\label{sec:CE-contingent}  
In Section \ref{sec:opt-sch}, we considered the optimal scheduling of energy from the perspective of a hypothetical social planner. We now show that the optimal scheduling decisions of the social planner can be realized through a spot market with contingent prices. In this market, the LSE is a buyer, and the aggregator is the seller. 

At time $t_0$, the LSE buys $q$ units of energy from the day-ahead market at a price $\pi^{da}$. At time $t_1$, the information state $s$  is revealed. Depending on this revelation, the LSE purchases $y_s$ units of energy curtailment from the aggregator, paying a price $\pi_s^{in}$. 
This is a {\em contingent price} as it depends on the realized information state $s$.  
At time $T$, the LSE receives wind energy $w$ and purchases the required balancing energy $(l-q-y_s-w)_{+}$ from the real-time market at a price $\pi^{rt}$. 
 The expected cost for the LSE as a function of the first-stage purchase $q$ is given by,
\begin{equation} 
\label{eq:J-LSE-market-1}
J^{LSE}(q) =  \pi^{da} q +  \mathbb{E}_{s}[\min_{y_{s}\geq 0} J^{LSE}_{s}(y_{s};q)]
\end{equation}
where $J^{LSE}_{s}(y_{s};q)$ is the second stage cost and is given by,
\begin{equation}
J^{LSE}_{s}(y_{s};q) = \pi_{s}^{in} y_{s} + \mathbb{E}_{w}[\overline{\pi}^{rt}_{s} (l-q-y_{s}-w)_{+}]
\end{equation}  
The optimal first and second-stage purchase decisions of the LSE are
$q^{LSE}$ and $y^{LSE}_{s}$ respectively. These are given by
\begin{align*}
\label{eq:J-LSE-q-ys}
\begin{cases}
 q^{LSE} = \arg \min_{q\geq 0} J^{LSE}(q) \\
 y^{LSE}_{s} = \min_{0\leq y_{s} \leq l} J^{LSE}_{s}(y_{s};q^{LSE})
 \end{cases}
\end{align*}
 
The expected cost for the aggregator under the information state $s$ is
\begin{equation}
\label{eq:J-agg-market}
J^{agg}_{s}(y_{s}) =   \phi(y_s) -  \pi_{s}^{in}y_{s}.
\end{equation}
The optimal selling decision of the aggregator is
\begin{equation*}
\label{eq:J-agg-market_loptselling}
 y^{agg}_{s} =  \min_{0\leq y_{s} \leq l} J^{agg}_{s}(y_{s})
\end{equation*}

Note that the optimal buying/selling decisions of agents (LSE/aggregator) depend on the contingent prices $\pi_s^{in}$. The market is said to be in {\em equilibrium} if the prices are such that the optimal buying and selling decisions of the agents are consistent under {{\em all} realizations of $s$. We make this notion more precise below.
 
\begin{definition}[Competitive Equilibrium with Contingent Prices]
The contingent prices $\{\pi_{s}^{*in}\}$, optimal buying decisions of the LSE $q^{*LSE}, \{y^{*LSE}_{s}\}$, optimal selling decisions of the aggregator $\{y^{*agg}_{s}\}$ constitute a competitive equilibrium, if the following holds for \textbf{all} $s\in S$:
\begin{subnumcases}{\label{eq:defineCEforspot}}
 J^{LSE}(q^{*LSE}) = \min_{q\geq 0} J^{LSE}(q)  \label{eq:defineCEforspota}\\
 J^{LSE}_{s}(y^{*LSE}_{s}) = \min_{0\leq y_{s}\leq l} J^{LSE}_{s}(y_{s};q^{*LSE}) \label{eq:defineCEforspotb} \\
  J^{agg}_{s}(y^{*agg}_{s}) = \min_{0\leq y_{s} \leq l} J^{agg}_{s}(y_{s}) \label{eq:defineCEforspotc} \\
  y^{*LSE}_{s}=y^{*agg}_{s} \label{eq:defineCEforspotd}
\end{subnumcases}
\end{definition}

Here (\ref{eq:defineCEforspota}) and (\ref{eq:defineCEforspotb}) require $(q^{*LSE},y^{*LSE}_{s})$ to be the optimal decision of the buyer,  (\ref{eq:defineCEforspotc}) requires $y^{*agg}_{s}$ to be the optimal decision of the seller, and  (\ref{eq:defineCEforspotd}) ensures that the traded demand response quantities are in balance. We require this balance at all realizations of $s$.

Let $J^{*LSE}$ be the expected cost for the LSE, and let $J^{*agg}$ be the expected cost for the aggregator at any competitive equilibrium. The system cost of the market at any competitive equilibrium is 
\begin{equation}
\label{eq:k1}
J^{*cp} = J^{*LSE} + J^{*agg}.
\end{equation} 
Define the minimum system cost for the social planner as $J^{*e}=J^e(q^e)$. This is a lower bound of the system cost for any market. Therefore, we can use  $J^{*e}$ as a benchmark to evaluate the {\em efficiency} of the options market. The market is called efficient (or socially optimal) if the system cost for the market attains the lower bound $J^{*e}$ at the competitive equilibrium. We make this precise in the following definition.

\begin{definition}[Socially Optimal Equilibrium with Contingent Prices] 
An equilibrium with contingent prices is said to be socially optimal, if $J^{*cp} = J^{*e}$.
\end{definition}  
We now offer the main result of this section. 

\begin{theorem}
\label{thm:ce-market} 
(a) There exists at least one competitive equilibrium under contingent pricing. 
(b) All competitive equilibria are socially optimal. Equivalently, define $y_{s}^*= y_{s}^{*LSE} = y_{s}^{*agg}$ at any competitive equilibrium, then
\begin{subnumcases}{\label{optimalconditionofspot}}
 J^{e}(q^{*LSE}) = \min_{q\geq 0} J^{e}(q),  \label{optimalconditionofspot1}\\
 J_{s}^e(y_{s}^*) \quad = \min_{0\leq y_{s} \leq l}J_{s}^e(y_{s};q^{*LSE}). \label{optimalconditionofspot2} 
\end{subnumcases}
\end{theorem}

The proof is deferred to Appendix C. Condition (\ref{optimalconditionofspot}) requires that competitive equilibrium is the optimal solution to the social planner's problem. Therefore, it can be computed by solving (\ref{eq:social planner-ys-1}) and (\ref{eq:social planner-qe-1}). This result implies that the optimal scheduling of the  social planner can be realized though an intermediate spot market with contingent prices. 
  
\section{Options Markets and Competitive Equilibrium}    
\label{sec:CE-option}  
In the previous section, we showed that the intermediate spot market is efficient. However, implementing intermediate spot markets requires organized infrastructure and regulatory approval which can be prohibitive. We now present an intermediate market for demand response using call options. These are private over-the-counter transactions which do not need utility blessing or organized infrastructure. 

\subsection{Options Market}
\label{options_market_original}
At time $t_0$, the LSE purchases energy $q$ in the day-ahead market. Concurrently, he buys  $x$ units of  options from the aggregator at the {\em option price} $\pi^o$. By purchasing these options, the LSE acquires the right, without the obligation, to receive $y$ units of load reduction from the aggregator where $0\leq y\leq x$. At time $t_1$, the LSE can {\em exercise} these options by paying a {\em strike price} $\pi^{sp}$ per contract. Clearly, the number of exercised options $y_{s}$ depends on the information state $s$ revealed at time $t_{1}$. The strike price $\pi^{sp}$ is {\em ex ante}, and does not depend on the information state. At time $T$, the aggregator delivers the contractually obligated load reduction $y_s$. The LSE observes the wind energy $w$ and purchases the remaining balancing energy $(l-q-y_{s}-w)_{+}$ in the real time market.

Since we are considering a competitive market, we assume the agents are rational and price takers. They make their buying/selling decisions based on the market prices $\pi^{o}$ and $\pi^{sp}$. The expected cost for the LSE is a function of the first stage decisions $q$ and $x$:
\begin{equation}
\label{eq:J-LSE-option-1}
\tilde{J}^{LSE}(q, x) = \pi^o x + \pi^{da} q  + \mathbb{E}_{s}[\min_{0\leq y_{s} \leq x} \tilde{J}^{LSE}_{s}(y_{s}) ], 
\end{equation}
Here $\tilde{J}^{LSE}_{s}(\cdot)$ is the second stage cost for the LSE given by
\begin{equation}
\label{eq:J-LSE-option-2}
\tilde{J}^{LSE}_{s}(y_{s}) = \pi^{sp} y_{s} + \mathbb{E}_{w}[\overline{\pi}^{rt}_{s} (l-q-y_{s}-w)_{+}]. 
\end{equation}
Denote the optimal first and second-stage decisions of the LSE by $\tilde{q}^{LSE}, \tilde{x}^{LSE}$ and $\tilde{y}^{LSE}_{s}$. These decisions solve 
\begin{align}
\begin{cases}
\label{eq:J-LSE-option-q-ys} 
(\tilde{q}^{LSE}, \tilde{x}^{LSE}) = \arg \min_{(q,x)} \tilde{J}^{LSE}(q, x),\nonumber \\ \tilde{y}^{LSE}_{s} = \arg \min_{y_s\leq x } \tilde{J}^{LSE}_{s}(y_s) 
\end{cases}
\end{align}
  
In the options market, the expected cost for the aggregator is
\begin{equation}
\label{eq:J-agg-option}
\tilde{J}^{agg}(x) =   \mathbb{E}_{s}[\phi(y_{s}) - \pi^{sp} y_{s}] -\pi^o x,
\end{equation}

The decision variable of the aggregator is the quantity of options $x$ offered for sale. The optimal selling decision is:
\begin{equation}
\label{eq:J-agg-option_solution}
 \tilde{x}^{agg} = \arg \min_{x\geq 0} \tilde{J}^{agg}(x)
\end{equation}

We now define an equilibrium notion for our options market.
\begin{definition}[Competitive Equilibrium for Options Market]
The options price $\pi^{*o}$, the strike price $\pi^{*sp}$, the optimal day-ahead purchase $\tilde{q}^*$, the optimal buying decision of the LSE $\tilde{x}^{*LSE}$ and the optimal selling decision of the aggregator $\tilde{x}^{*agg}$ constitute a competitive equilibrium if
\begin{align*}
\begin{cases}
 \tilde{J}^{LSE}(\tilde{q}^*,\tilde{x}^{*LSE}) = \min_{(q,x)} \tilde{J}^{LSE}(q,x) \\
\tilde{J}_s^{LSE}(\tilde{y}_s^{*LSE}) = \min_{0\leq y_s \leq \tilde{x}^{*LSE}} \tilde{J}_s^{LSE}(y_s) \\
\tilde{J}^{agg}(\tilde{x}^{*agg}) = \min_{0\leq x \leq l} \tilde{J}^{agg}(x) \\
 \tilde{x}^{*LSE}=\tilde{x}^{*agg}
\end{cases}
\end{align*} 
\end{definition} 

At the competitive equilibrium, the volume of options that the LSE is willing to buy balances the volume  of options that the aggregator is willing to sell. Therefore, we have $\tilde{x}^{LSE} = \tilde{x}^{agg}$.  
We now offer the main results of this section. 
\begin{theorem} 
\label{thm:ce-option_oldversion} 
There exists a competitive equilibrium for options market. Define the following at a competitive equilibrium, $\tilde{q}^{*} = \tilde{q}^{LSE}, \tilde{x}^{*} = \tilde{x}^{*LSE} = \tilde{x}^{*agg}$ and $\tilde{y}^{*}_{s} = \tilde{y}^{*LSE}_{s}$. Then the competitive equilibrium satisfies,
\begin{align*}
\pi^{da} &-\mathbb{E}_{s}[\overline{\pi}^{rt}_{s} P_{s}(l-\tilde{q}^{*}-\tilde{y}^{*}_{s})] = 0 \\
\pi^{*o} &+ \pi^{*LSE} \mathbb{E}_{s}[\mathbb{I}\{\tilde{y}^{*}_{s} = \tilde{x}^{*}\}] \nonumber  \\
& - \mathbb{E}_{s}[\overline{\pi}^{rt}_{s} P_{s}(l-\tilde{q}^{*}-\tilde{y}^{*}_{s})  \mathbb{I}\{\tilde{y}^{*}_{s} = \tilde{x}^{*}\}] = 0 \\
\pi^{*o} &+ \pi^{*LSE} \mathbb{E}_{s}[\mathbb{I}\{\tilde{y}^{*}_{s} = x^{*}\}] - \phi'(x^{*}) \mathbb{E}_{s}[\mathbb{I}\{\tilde{y}^{*}_{s} = \tilde{x}^{*}\}] = 0, 
\end{align*}
where $\tilde{y}_s^*$ satisfies,
\begin{align*}
\tilde{y}^*_{s} &= \left\lbrace \begin{array}{ccc}
0, \quad \quad  \quad \quad \quad \quad \quad \quad \text{if}~ P_{s}(l-\tilde{q}^*) < \pi^{*sp}/\bar{\pi}^{rt}_{s}  \\
\tilde{x}^{*},   \quad \quad \quad \quad \quad \quad  \text{if}~P_{s}(l-\tilde{q}^*-\tilde{x}^*) > \pi^{*sp}/\bar{\pi}^{rt}_{s} \\
l - \tilde{q}^*- P^{-1}_{s}(\pi^{*sp}/\bar{\pi}^{rt}_{s}), \quad \quad \quad \quad \quad  \textnormal{otherwise}
\end{array} \right.
\end{align*}
\end{theorem} 

The proof of this theorem is given in Appendix D. We comment that the competitive equilibrium consists of four variables $(\pi^{*o},\pi^{*sp},\tilde{q}^*,\tilde{x}^*)$ determined by three equations. Therefore, there is one degree of freedom which induces multiple competitive equilibria. We will illustrate these equilibria prices through a numerical simulation in Section \ref{simulation_sec}.

\subsection{Redesign of Options Market}
\label{options_market_redesign}

The options market proposed in the previous section is asymmetric with respect to the decision of the buyer and the seller. That is the decision of the LSE is $q$ and $x$, while the decision of the aggregator is only $x$. This asymmetry can provide market advantage to the buyer. To address this concern, we propose a redesign of the options market where the decision of the buyer and the seller is symmetric. We show the existence of a competitive equilibrium, and study its various properties.

\noindent B.1 {\em Symmetric Decision Making}

Consider the following modification to the options market: before time $t_0$, the aggregator proposes a demand response offer to the LSE. The aggregator chooses $l'>0$ and dictates that $x+q=l'$. This endows the aggregator the power to negotiate on $q$: the aggregator offers $x$ units of options, only if the LSE buys $l'-x$ units of energy in the day-ahead market. For the moment, we treat $l'$ as given. 

Upon receiving the demand response offer, the LSE decides whether or not to accept it. There is no trade of load reduction if the offer is not accepted. When the offer is accepted, the expected cost for the LSE is
\begin{equation}
\label{eq:J-LSE-option-1_redesign}
\tilde{J}^{LSE}(x) = \pi^o x + \pi^{da} (l'-x)  + \mathbb{E}_{s}[\min_{y_{s}, y_{s} \leq x} \tilde{J}^{LSE}_{s}(y_{s};x) ], 
\end{equation}
where $\tilde{J}_s^{LSE}(\cdot)$ is the second stage cost and is given by, 
\begin{equation}
\label{eq:J-LSE-option-2_redesign}
\tilde{J}^{LSE}_{s}(y_{s};x) = \pi^{sp} y_{s} + \mathbb{E}_{w}[\overline{\pi}^{rt}_{s} (l-l'+x-y_{s}-w)_{+}]. 
\end{equation}
The optimal first and second-stage decisions of the LSE are 
\begin{align}
\label{eq:J-LSE-option-q-ys_redesign} 
\begin{cases}
\tilde{x}^{LSE} = \arg \min_{x\geq 0} \tilde{J}^{LSE}(x), \\ \tilde{y}^{LSE}_{s}(x) = \arg \min_{y_s\leq x} \tilde{J}^{LSE}_{s}(y_s; x) 
\end{cases}
\end{align}
Note that the second-stage decision $\tilde{y}_s^{LSE}$ depends on $x$. From now on, we do not express this dependence explicitly as it is implied by  context. The   expected cost for the aggregator and its optimal decisions remain as in (\ref{eq:J-agg-option}) and (\ref{eq:J-agg-option_solution}).  

We assume that the LSE and the aggregator are price takers. The options market attains a competitive equilibrium if the supply of options balances the demand of options. 

\begin{definition}[Competitive Equilibrium for Options Market]
Given any $l'$ such that $0\leq l' \leq l$, the options price $\pi^{*o}$, the strike price $\pi^{*sp}$, the optimal buying decision of the LSE $\tilde{x}^{*LSE}$ and the optimal selling decision of the aggregator $\tilde{x}^{*agg}$ constitute a competitive equilibrium if:
\begin{subnumcases}{\label{competitive_opt}}
 \tilde{J}^{LSE}(\tilde{x}^{*LSE}) = \min_{0\leq x\leq l'} \tilde{J}^{LSE}(x),  \label{competitive_opta}\\
\tilde{J}_s^{LSE}(\tilde{y}_s^{*LSE}) = \min_{0\leq y_s \leq x} \tilde{J}_s^{LSE}(y_s;x^{*LSE}) \label{competitive_optb} \\
\tilde{J}^{agg}(\tilde{x}^{*agg}) = \min_{0\leq x \leq l'} \tilde{J}^{agg}(x) \label{competitive_optc} \\
  \tilde{x}^{*LSE} = \tilde{x}^{*agg}. \label{competitive_optd}
\end{subnumcases}
\end{definition}

The choice of $l'$ is determined by the willingness of the LSE to accept the demand response offer. If the LSE accepts the offer, its optimal cost is $\tilde{J}^{LSE}(\tilde{x}^{LSE})$. Else, its cost is equal to that of optimal cost without DR, i.e., $J_{ndr}^{*e}$. Thus,  the LSE will accept the contract proposed by the aggregator if
\begin{equation}
\label{participationcons}
J_{ndr}^{*e} \geq \tilde{J}^{LSE}(\tilde{x}^{LSE}).
\end{equation} 
However, $\tilde{J}^{LSE}(x^{LSE})$ depends on the options price $\pi^o$, which is not revealed when the LSE makes the decision. Ideally,  $l'$ should be such that (\ref{participationcons}) holds for any $\pi^o$. We present a candidate of $l^{'}$ that satisfies this condition:

\begin{proposition}
\label{howtochoosecontract}
If $l'=q_{ndr}^e$, the LSE always accepts the demand response offer, i.e., $J_{ndr}^{*e}\geq \tilde{J}^{LSE}(x^{LSE})$ for $\forall \pi^o\geq 0$.
\end{proposition} 

The idea is as follows: $q=q_{ndr}^e$ is the optimal decision of the LSE if it declines the demand response offer. Therefore, when $l'=q_{ndr}^e$, the LSE loses nothing if it accepts the demand response offer. This is because, there exists a LSE decision, i.e., $x = 0$ and $q = q_{ndr}^e$, that satisfies the condition $x+q=q_{ndr}^e$ and also attains the same cost. 

\vspace{0.3cm}
\noindent B.2 {\em Properties of Competitive Equilibrium}
\vspace{0.1cm}

We now focus on the existence, efficiency and optimality of the competitive equilibrium for options market. 
\begin{theorem} 
\label{thm:ce-option} 
Given any $l'\in [0,l]$, there exists a competitive equilibrium $(\pi^{*o}, \pi^{*sp}, \tilde{x}^{*LSE}, \tilde{x}^{*agg})$ for the options market, and $\tilde{x}^{*LSE} = \tilde{x}^{*agg}$ is the optimal solution to: 
\begin{equation}
\label{eq:ce-option-ys}
\min_{0\leq x \leq l'}  \pi^{da} (l'-x)+ \mathbb{E}_{s} [\phi(\tilde{y}_s^{LSE})+\tilde{J}_s^{LSE}(\tilde{y}_s^{LSE})],
\end{equation}
where $\tilde{y}_s^{LSE}$ is the second stage optimal decision for the LSE and,
\begin{align}
\label{optimaLSEcondstagedec}
\tilde{y}^{LSE}_{s} &= \left\lbrace \begin{array}{ccc}
0, \quad \quad  \quad \quad \quad \quad \quad \quad \text{if}~ P_{s}(l-l'+x) < \pi^{*sp}/\bar{\pi}^{rt}_{s}  \\
x, \quad \quad \quad \quad \quad \quad \quad \quad \quad \quad  \text{if}~P_{s}(l-l') > \pi^{*sp}/\bar{\pi}^{rt}_{s} \\
l - l'+x- P^{-1}_{s}(\pi^{*sp}/\bar{\pi}^{rt}_{s}), \quad \quad \quad \quad \quad  \textnormal{otherwise}
\end{array} \right.
\end{align}
\end{theorem} 

The proof is given in Appendix E. The optimization problem (\ref{eq:ce-option-ys}) is convex, and the optimal value of (\ref{eq:ce-option-ys}) is the social cost at the competitive equilibrium of the options market.

Similar to the options market in Section \ref{options_market_original}, there are multiple competitive equilibria because for any equilibrium price pair $\pi^{s*o}$ and $\pi^{*sp}$, a higher options price with a lower strike price can be equally acceptable to both the LSE and the aggregator. 
To compare the efficiency of different markets, let $\tilde{J}^{*LSE}(\pi^{sp})$ and $\tilde{J}^{*agg}(\pi^{sp})$ be the expected cost at competitive equilibrium for the LSE and the aggregator, respectively. Define the system cost at competitive equilibrium by,
\begin{equation}
\label{eq:k1}
\tilde{J}^{*cp}(\pi^{sp}) = \tilde{J}^{*LSE}(\pi^{sp})+ \tilde{J}^{*agg}(\pi^{sp}).
\end{equation} 
In addition, let $J_{ndr}^{*e}$ be the optimal value of  problem (\ref{eq:social planner-J-ndr-1}), then the following proposition provides a comparison of the optimal cost of the different markets that we have discussed so far.
\begin{proposition}
\label{efficiencyofoptions}
Given any $l'$ and $\pi^{sp}$, the social cost of the options  market at the competitive equilibrium is lower bounded by $J^{*cp}$ and upper bounded by $J_{ndr}^{*e}$, i.e., $J^{*cp}\leq \tilde{J}^{*cp}(\pi^{sp})  \leq J_{ndr}^{*e}$.
\end{proposition}

The proof of Proposition \ref{efficiencyofoptions} is given in the Appendix. It indicates that the efficiency of the options market outperforms that of the market without demand response, but is no better than that of the spot market with contingent prices. 

The following theorem presents the optimal strike price that minimizes the social cost at the competitive equilibrium:
\begin{theorem}
\label{optimalstrikeprice}
There exists an optimal strike price $\tilde{\pi}^{*LSE}$, such that $\tilde{J}^{*cp}(\tilde{\pi}^{*LSE})\leq \tilde{J}^{*cp}(\pi^{sp})$ for all $\pi^{sp}$, and 
$\tilde{\pi}^{*LSE}$ satisfies:
\begin{equation}
\tilde{\pi}^{*LSE}=\dfrac{\int_{s_1}^{s_2} \phi'(y_s) \beta(s)ds  }{\int_{s_1}^{s_2}\beta(s)ds }  
\end{equation}
where $\beta(s)=\dfrac{\alpha(s)}{\bar{\pi}_s^{rt}p_s(l-q-y_s)}$.
\end{theorem}
The proof of Theorem \ref{optimalstrikeprice} is in the Appendix. It shows that the optimal strike price is the average of the marginal disutility over a skewed distribution $\beta(s)$. 

\begin{figure*}[bt]%
\begin{minipage}[b]{0.32\linewidth}
\centering
%
%
\definecolor{mycolor1}{rgb}{0.00000,0.44700,0.74100}%
\begin{tikzpicture}

\begin{axis}[%
width=1.8in,
height=1in,
at={(0.0in,0.0in)},
scale only axis,
bar shift auto,
xmin=-0.0139999999999999,
xmax=1.014,
xlabel style={font=\color{white!15!black}},
xlabel={s},
ymin=0,
ymax=0.08,
ylabel style={font=\color{white!15!black}},
ylabel={histogram},
axis background/.style={fill=white},
legend style={legend cell align=left, align=left, draw=white!15!black}
]
\addplot[ybar, bar width=0.016, fill=mycolor1, draw=black, area legend] table[row sep=crcr] {%
0.01	0.0340330140615447\\
0.03	0.0690849806399022\\
0.05	0.0729570002037905\\
0.07	0.079885877318117\\
0.09	0.0590992459751376\\
0.11	0.0383126146321581\\
0.13	0.0342368045649073\\
0.15	0.0319951090279193\\
0.17	0.0340330140615447\\
0.19	0.0224169553698798\\
0.21	0.0262889749337681\\
0.23	0.0240472793967801\\
0.25	0.023639698390055\\
0.27	0.0234359078866925\\
0.29	0.0216017933564296\\
0.31	0.0205828408396169\\
0.33	0.0203790503362543\\
0.35	0.0218055838597921\\
0.37	0.0173221927858162\\
0.39	0.0201752598328918\\
0.41	0.0179335642959038\\
0.43	0.0185449358059914\\
0.45	0.0175259832891787\\
0.47	0.0191563073160791\\
0.49	0.0252700224169554\\
0.51	0.0189525168127165\\
0.53	0.0201752598328918\\
0.55	0.0169146117790911\\
0.57	0.0148767067454657\\
0.59	0.0195638883228042\\
0.61	0.0150804972488282\\
0.63	0.0134501732219279\\
0.65	0.0128388017118402\\
0.67	0.0112084776849399\\
0.69	0.011616058691665\\
0.71	0.0108008966782148\\
0.73	0.00774403912777665\\
0.75	0.0061137151008763\\
0.77	0.00835541063786428\\
0.79	0.0061137151008763\\
0.81	0.00590992459751376\\
0.83	0.00346443855716324\\
0.85	0.00163032402690035\\
0.87	0.000815162013450173\\
0.89	0.000203790503362543\\
0.91	0.000407581006725087\\
0.93	0\\
0.95	0\\
0.97	0\\
0.99	0\\
};
\addplot[forget plot, color=white!15!black] table[row sep=crcr] {%
-0.0139999999999999	0\\
1.014	0\\
};

\end{axis}

\end{tikzpicture}%
\caption{The histogram of information state $s$ based on real wind data. }
\label{PDF_s}
\end{minipage}
\begin{minipage}[b]{0.01\linewidth}
\hfill
\end{minipage}
\begin{minipage}[b]{0.32\linewidth}
\centering
\input{figure2_2}
\caption{Empirical distribution of  $\sigma$ and it analytic approximation } 
\label{CDF_s}
\end{minipage}
\begin{minipage}[b]{0.01\linewidth}
\hfill
\end{minipage}
\begin{minipage}[b]{0.32\linewidth}
\centering
%
%
\definecolor{mycolor1}{rgb}{0.60000,0.20000,0.00000}
\begin{tikzpicture}

\begin{axis}[%
width=1.8in,
height=1in,
at={(2.167in,1.454in)},
scale only axis,
xmin=0,
xmax=1,
xlabel style={font=\color{white!15!black}},
xlabel={Information State},
ymin=0,
ymax=0.6,
ylabel style={font=\color{white!15!black}},
ylabel={$y_s$(MWh)},
axis background/.style={fill=white},
legend style={legend cell align=left, align=left, draw=white!15!black}
]
\addplot [color=mycolor1, line width=1pt, forget plot]
  table[row sep=crcr]{%
0	0.556952819824219\\
0.00333333333333333	0.556621398925781\\
0.00666666666666667	0.556289978027344\\
0.01	0.555958557128906\\
0.0133333333333333	0.555627136230469\\
0.0166666666666667	0.555295715332031\\
0.02	0.554853820800781\\
0.0233333333333333	0.554522399902344\\
0.0266666666666667	0.554190979003906\\
0.03	0.553859558105469\\
0.0333333333333333	0.553528137207031\\
0.0366666666666667	0.553086242675781\\
0.04	0.552754821777344\\
0.0433333333333333	0.552423400878906\\
0.0466666666666667	0.552091979980469\\
0.05	0.551760559082031\\
0.0533333333333333	0.551318664550781\\
0.0566666666666667	0.550987243652344\\
0.06	0.550655822753906\\
0.0633333333333333	0.550324401855469\\
0.0666666666666667	0.549992980957031\\
0.07	0.549551086425781\\
0.0733333333333333	0.549219665527344\\
0.0766666666666667	0.548888244628906\\
0.08	0.548556823730469\\
0.0833333333333333	0.548225402832031\\
0.0866666666666667	0.547893981933594\\
0.09	0.547452087402344\\
0.0933333333333333	0.547120666503906\\
0.0966666666666667	0.546789245605469\\
0.1	0.546457824707031\\
0.103333333333333	0.546126403808594\\
0.106666666666667	0.545684509277344\\
0.11	0.545353088378906\\
0.113333333333333	0.545132141113281\\
0.116666666666667	0.544690246582031\\
0.12	0.544358825683594\\
0.123333333333333	0.543916931152344\\
0.126666666666667	0.543585510253906\\
0.13	0.543254089355469\\
0.133333333333333	0.542922668457031\\
0.136666666666667	0.542591247558594\\
0.14	0.542259826660156\\
0.143333333333333	0.541817932128906\\
0.146666666666667	0.541486511230469\\
0.15	0.541155090332031\\
0.153333333333333	0.540823669433594\\
0.156666666666667	0.540492248535156\\
0.16	0.540050354003906\\
0.163333333333333	0.539718933105469\\
0.166666666666667	0.539387512207031\\
0.17	0.539056091308594\\
0.173333333333333	0.538724670410156\\
0.176666666666667	0.538282775878906\\
0.18	0.537951354980469\\
0.183333333333333	0.537619934082031\\
0.186666666666667	0.537288513183594\\
0.19	0.536846618652344\\
0.193333333333333	0.536515197753906\\
0.196666666666667	0.536183776855469\\
0.2	0.535741882324219\\
0.203333333333333	0.535410461425781\\
0.206666666666667	0.535079040527344\\
0.21	0.534637145996094\\
0.213333333333333	0.534305725097656\\
0.216666666666667	0.533863830566406\\
0.22	0.533532409667969\\
0.223333333333333	0.533090515136719\\
0.226666666666667	0.532648620605469\\
0.23	0.532206726074219\\
0.233333333333333	0.531764831542969\\
0.236666666666667	0.531322937011719\\
0.24	0.530770568847656\\
0.243333333333333	0.530328674316406\\
0.246666666666667	0.529776306152344\\
0.25	0.529223937988281\\
0.253333333333333	0.528561096191406\\
0.256666666666667	0.528008728027344\\
0.26	0.527345886230469\\
0.263333333333333	0.526572570800781\\
0.266666666666667	0.525799255371094\\
0.27	0.524915466308594\\
0.273333333333333	0.524031677246094\\
0.276666666666667	0.523037414550781\\
0.28	0.521932678222656\\
0.283333333333333	0.520827941894531\\
0.286666666666667	0.519502258300781\\
0.29	0.518176574707031\\
0.293333333333333	0.516740417480469\\
0.296666666666667	0.515193786621094\\
0.3	0.513536682128906\\
0.303333333333333	0.511769104003906\\
0.306666666666667	0.509780578613281\\
0.31	0.507792053222656\\
0.313333333333333	0.505582580566406\\
0.316666666666667	0.503262634277344\\
0.32	0.500832214355469\\
0.323333333333333	0.498291320800781\\
0.326666666666667	0.495639953613281\\
0.33	0.492878112792969\\
0.333333333333333	0.490005798339844\\
0.336666666666667	0.486912536621094\\
0.34	0.483819274902344\\
0.343333333333333	0.480505065917969\\
0.346666666666667	0.477190856933594\\
0.35	0.473655700683594\\
0.353333333333333	0.470120544433594\\
0.356666666666667	0.466474914550781\\
0.36	0.462718811035156\\
0.363333333333333	0.458852233886719\\
0.366666666666667	0.454875183105469\\
0.37	0.450898132324219\\
0.373333333333333	0.446810607910156\\
0.376666666666667	0.442723083496094\\
0.38	0.438414611816406\\
0.383333333333333	0.434106140136719\\
0.386666666666667	0.429797668457031\\
0.39	0.425378723144531\\
0.393333333333333	0.420959777832031\\
0.396666666666667	0.416430358886719\\
0.4	0.411790466308594\\
0.403333333333333	0.407150573730469\\
0.406666666666667	0.402510681152344\\
0.41	0.397760314941406\\
0.413333333333333	0.393009948730469\\
0.416666666666667	0.388149108886719\\
0.42	0.383398742675781\\
0.423333333333333	0.378427429199219\\
0.426666666666667	0.373566589355469\\
0.43	0.368595275878906\\
0.433333333333333	0.363623962402344\\
0.436666666666667	0.358652648925781\\
0.44	0.353570861816406\\
0.443333333333333	0.348489074707031\\
0.446666666666667	0.343407287597656\\
0.45	0.338325500488281\\
0.453333333333333	0.333133239746094\\
0.456666666666667	0.327940979003906\\
0.46	0.322748718261719\\
0.463333333333333	0.317556457519531\\
0.466666666666667	0.312364196777344\\
0.47	0.307061462402344\\
0.473333333333333	0.301869201660156\\
0.476666666666667	0.296566467285156\\
0.48	0.291263732910156\\
0.483333333333333	0.285960998535156\\
0.486666666666667	0.280547790527344\\
0.49	0.275245056152344\\
0.493333333333333	0.269831848144531\\
0.496666666666667	0.264418640136719\\
0.5	0.259115905761719\\
0.503333333333333	0.253702697753906\\
0.506666666666667	0.248179016113281\\
0.51	0.242765808105469\\
0.513333333333333	0.237352600097656\\
0.516666666666667	0.231939392089844\\
0.52	0.226415710449219\\
0.523333333333333	0.220892028808594\\
0.526666666666667	0.215478820800781\\
0.53	0.209955139160156\\
0.533333333333333	0.204431457519531\\
0.536666666666667	0.198907775878906\\
0.54	0.193384094238281\\
0.543333333333333	0.187860412597656\\
0.546666666666667	0.182336730957031\\
0.55	0.176813049316406\\
0.553333333333333	0.171178894042969\\
0.556666666666667	0.165655212402344\\
0.56	0.160021057128906\\
0.563333333333333	0.154497375488281\\
0.566666666666667	0.148863220214844\\
0.57	0.143339538574219\\
0.573333333333333	0.137705383300781\\
0.576666666666667	0.132071228027344\\
0.58	0.126547546386719\\
0.583333333333333	0.120913391113281\\
0.586666666666667	0.115279235839844\\
0.59	0.109645080566406\\
0.593333333333333	0.104010925292969\\
0.596666666666667	0.0983767700195313\\
0.6	0.0927426147460938\\
0.603333333333333	0.0871084594726562\\
0.606666666666667	0.0814743041992187\\
0.61	0.0757296752929688\\
0.613333333333333	0.0700955200195312\\
0.616666666666667	0.0644613647460937\\
0.62	0.0588272094726563\\
0.623333333333333	0.0530825805664062\\
0.626666666666667	0.0474484252929687\\
0.63	0.0418142700195313\\
0.633333333333333	0.0360696411132813\\
0.636666666666667	0.0304354858398437\\
0.64	0.0246908569335937\\
0.643333333333333	0.0190567016601563\\
0.646666666666667	0.0133120727539062\\
0.65	0.00767791748046875\\
0.653333333333333	0.00193328857421875\\
0.656666666666667	0\\
0.66	0\\
0.663333333333333	0\\
0.666666666666667	0\\
0.67	0\\
0.673333333333333	0\\
0.676666666666667	0\\
0.68	0\\
0.683333333333333	0\\
0.686666666666667	0\\
0.69	0\\
0.693333333333333	0\\
0.696666666666667	0\\
0.7	0\\
0.703333333333333	0\\
0.706666666666667	0\\
0.71	0\\
0.713333333333333	0\\
0.716666666666667	0\\
0.72	0\\
0.723333333333333	0\\
0.726666666666667	0\\
0.73	0\\
0.733333333333333	0\\
0.736666666666667	0\\
0.74	0\\
0.743333333333333	0\\
0.746666666666667	0\\
0.75	0\\
0.753333333333333	0\\
0.756666666666667	0\\
0.76	0\\
0.763333333333333	0\\
0.766666666666667	0\\
0.77	0\\
0.773333333333333	0\\
0.776666666666667	0\\
0.78	0\\
0.783333333333333	0\\
0.786666666666667	0\\
0.79	0\\
0.793333333333333	0\\
0.796666666666667	0\\
0.8	0\\
0.803333333333333	0\\
0.806666666666667	0\\
0.81	0\\
0.813333333333333	0\\
0.816666666666667	0\\
0.82	0\\
0.823333333333333	0\\
0.826666666666667	0\\
0.83	0\\
0.833333333333333	0\\
0.836666666666667	0\\
0.84	0\\
0.843333333333333	0\\
0.846666666666667	0\\
0.85	0\\
0.853333333333333	0\\
0.856666666666667	0\\
0.86	0\\
0.863333333333333	0\\
0.866666666666667	0\\
0.87	0\\
0.873333333333333	0\\
0.876666666666667	0\\
0.88	0\\
0.883333333333333	0\\
0.886666666666667	0\\
0.89	0\\
0.893333333333333	0\\
0.896666666666667	0\\
0.9	0\\
0.903333333333333	0\\
0.906666666666667	0\\
0.91	0\\
0.913333333333333	0\\
0.916666666666667	0\\
0.92	0\\
0.923333333333333	0\\
0.926666666666667	0\\
0.93	0\\
0.933333333333333	0\\
0.936666666666667	0\\
0.94	0\\
0.943333333333333	0\\
0.946666666666667	0\\
0.95	0\\
0.953333333333333	0\\
0.956666666666667	0\\
0.96	0\\
0.963333333333333	0\\
0.966666666666667	0\\
0.97	0\\
0.973333333333333	0\\
0.976666666666667	0\\
0.98	0\\
0.983333333333333	0\\
0.986666666666667	0\\
0.99	0\\
0.993333333333333	0\\
0.996666666666667	0\\
1	0\\
};

\end{axis}
\end{tikzpicture}%
\caption{Load reduction called at the spot market under different information state. }
\label{fig:loadreductionsopt}
\end{minipage}
\end{figure*}

\begin{figure*}[bt]%
\begin{minipage}[b]{0.32\linewidth}
\centering
%
%
\definecolor{mycolor1}{rgb}{0.60000,0.20000,0.00000}
\begin{tikzpicture}

\begin{axis}[%
width=1.8in,
height=1in,
at={(1.081in,1.067in)},
scale only axis,
xmin=0,
xmax=1,
xlabel style={font=\color{white!15!black}},
xlabel={Information State},
ymin=14,
ymax=32,
ylabel style={font=\color{white!15!black}},
ylabel={$p_s^{in}$ (\$/MWh)},
axis background/.style={fill=white},
legend style={legend cell align=left, align=left, draw=white!15!black}
]
\addplot [color=mycolor1, line width=1pt, forget plot]
  table[row sep=crcr]{%
0	31.7085845947266\\
0.00333333333333333	31.6986419677734\\
0.00666666666666667	31.6886993408203\\
0.01	31.6787567138672\\
0.0133333333333333	31.6688140869141\\
0.0166666666666667	31.6588714599609\\
0.02	31.6456146240234\\
0.0233333333333333	31.6356719970703\\
0.0266666666666667	31.6257293701172\\
0.03	31.6157867431641\\
0.0333333333333333	31.6058441162109\\
0.0366666666666667	31.5925872802734\\
0.04	31.5826446533203\\
0.0433333333333333	31.5727020263672\\
0.0466666666666667	31.5627593994141\\
0.05	31.5528167724609\\
0.0533333333333333	31.5395599365234\\
0.0566666666666667	31.5296173095703\\
0.06	31.5196746826172\\
0.0633333333333333	31.5097320556641\\
0.0666666666666667	31.4997894287109\\
0.07	31.4865325927734\\
0.0733333333333333	31.4765899658203\\
0.0766666666666667	31.4666473388672\\
0.08	31.4567047119141\\
0.0833333333333333	31.4467620849609\\
0.0866666666666667	31.4368194580078\\
0.09	31.4235626220703\\
0.0933333333333333	31.4136199951172\\
0.0966666666666667	31.4036773681641\\
0.1	31.3937347412109\\
0.103333333333333	31.3837921142578\\
0.106666666666667	31.3705352783203\\
0.11	31.3605926513672\\
0.113333333333333	31.3539642333984\\
0.116666666666667	31.3407073974609\\
0.12	31.3307647705078\\
0.123333333333333	31.3175079345703\\
0.126666666666667	31.3075653076172\\
0.13	31.2976226806641\\
0.133333333333333	31.2876800537109\\
0.136666666666667	31.2777374267578\\
0.14	31.2677947998047\\
0.143333333333333	31.2545379638672\\
0.146666666666667	31.2445953369141\\
0.15	31.2346527099609\\
0.153333333333333	31.2247100830078\\
0.156666666666667	31.2147674560547\\
0.16	31.2015106201172\\
0.163333333333333	31.1915679931641\\
0.166666666666667	31.1816253662109\\
0.17	31.1716827392578\\
0.173333333333333	31.1617401123047\\
0.176666666666667	31.1484832763672\\
0.18	31.1385406494141\\
0.183333333333333	31.1285980224609\\
0.186666666666667	31.1186553955078\\
0.19	31.1053985595703\\
0.193333333333333	31.0954559326172\\
0.196666666666667	31.0855133056641\\
0.2	31.0722564697266\\
0.203333333333333	31.0623138427734\\
0.206666666666667	31.0523712158203\\
0.21	31.0391143798828\\
0.213333333333333	31.0291717529297\\
0.216666666666667	31.0159149169922\\
0.22	31.0059722900391\\
0.223333333333333	30.9927154541016\\
0.226666666666667	30.9794586181641\\
0.23	30.9662017822266\\
0.233333333333333	30.9529449462891\\
0.236666666666667	30.9396881103516\\
0.24	30.9231170654297\\
0.243333333333333	30.9098602294922\\
0.246666666666667	30.8932891845703\\
0.25	30.8767181396484\\
0.253333333333333	30.8568328857422\\
0.256666666666667	30.8402618408203\\
0.26	30.8203765869141\\
0.263333333333333	30.7971771240234\\
0.266666666666667	30.7739776611328\\
0.27	30.7474639892578\\
0.273333333333333	30.7209503173828\\
0.276666666666667	30.6911224365234\\
0.28	30.6579803466797\\
0.283333333333333	30.6248382568359\\
0.286666666666667	30.5850677490234\\
0.29	30.5452972412109\\
0.293333333333333	30.5022125244141\\
0.296666666666667	30.4558135986328\\
0.3	30.4061004638672\\
0.303333333333333	30.3530731201172\\
0.306666666666667	30.2934173583984\\
0.31	30.2337615966797\\
0.313333333333333	30.1674774169922\\
0.316666666666667	30.0978790283203\\
0.32	30.0249664306641\\
0.323333333333333	29.9487396240234\\
0.326666666666667	29.8691986083984\\
0.33	29.7863433837891\\
0.333333333333333	29.7001739501953\\
0.336666666666667	29.6073760986328\\
0.34	29.5145782470703\\
0.343333333333333	29.4151519775391\\
0.346666666666667	29.3157257080078\\
0.35	29.2096710205078\\
0.353333333333333	29.1036163330078\\
0.356666666666667	28.9942474365234\\
0.36	28.8815643310547\\
0.363333333333333	28.7655670166016\\
0.366666666666667	28.6462554931641\\
0.37	28.5269439697266\\
0.373333333333333	28.4043182373047\\
0.376666666666667	28.2816925048828\\
0.38	28.1524383544922\\
0.383333333333333	28.0231842041016\\
0.386666666666667	27.8939300537109\\
0.39	27.7613616943359\\
0.393333333333333	27.6287933349609\\
0.396666666666667	27.4929107666016\\
0.4	27.3537139892578\\
0.403333333333333	27.2145172119141\\
0.406666666666667	27.0753204345703\\
0.41	26.9328094482422\\
0.413333333333333	26.7902984619141\\
0.416666666666667	26.6444732666016\\
0.42	26.5019622802734\\
0.423333333333333	26.3528228759766\\
0.426666666666667	26.2069976806641\\
0.43	26.0578582763672\\
0.433333333333333	25.9087188720703\\
0.436666666666667	25.7595794677734\\
0.44	25.6071258544922\\
0.443333333333333	25.4546722412109\\
0.446666666666667	25.3022186279297\\
0.45	25.1497650146484\\
0.453333333333333	24.9939971923828\\
0.456666666666667	24.8382293701172\\
0.46	24.6824615478516\\
0.463333333333333	24.5266937255859\\
0.466666666666667	24.3709259033203\\
0.47	24.2118438720703\\
0.473333333333333	24.0560760498047\\
0.476666666666667	23.8969940185547\\
0.48	23.7379119873047\\
0.483333333333333	23.5788299560547\\
0.486666666666667	23.4164337158203\\
0.49	23.2573516845703\\
0.493333333333333	23.0949554443359\\
0.496666666666667	22.9325592041016\\
0.5	22.7734771728516\\
0.503333333333333	22.6110809326172\\
0.506666666666667	22.4453704833984\\
0.51	22.2829742431641\\
0.513333333333333	22.1205780029297\\
0.516666666666667	21.9581817626953\\
0.52	21.7924713134766\\
0.523333333333333	21.6267608642578\\
0.526666666666667	21.4643646240234\\
0.53	21.2986541748047\\
0.533333333333333	21.1329437255859\\
0.536666666666667	20.9672332763672\\
0.54	20.8015228271484\\
0.543333333333333	20.6358123779297\\
0.546666666666667	20.4701019287109\\
0.55	20.3043914794922\\
0.553333333333333	20.1353668212891\\
0.556666666666667	19.9696563720703\\
0.56	19.8006317138672\\
0.563333333333333	19.6349212646484\\
0.566666666666667	19.4658966064453\\
0.57	19.3001861572266\\
0.573333333333333	19.1311614990234\\
0.576666666666667	18.9621368408203\\
0.58	18.7964263916016\\
0.583333333333333	18.6274017333984\\
0.586666666666667	18.4583770751953\\
0.59	18.2893524169922\\
0.593333333333333	18.1203277587891\\
0.596666666666667	17.9513031005859\\
0.6	17.7822784423828\\
0.603333333333333	17.6132537841797\\
0.606666666666667	17.4442291259766\\
0.61	17.2718902587891\\
0.613333333333333	17.1028656005859\\
0.616666666666667	16.9338409423828\\
0.62	16.7648162841797\\
0.623333333333333	16.5924774169922\\
0.626666666666667	16.4234527587891\\
0.63	16.2544281005859\\
0.633333333333333	16.0820892333984\\
0.636666666666667	15.9130645751953\\
0.64	15.7407257080078\\
0.643333333333333	15.5717010498047\\
0.646666666666667	15.3993621826172\\
0.65	15.2303375244141\\
0.653333333333333	15.0579986572266\\
0.656666666666667	15\\
0.66	15\\
0.663333333333333	15\\
0.666666666666667	15\\
0.67	15\\
0.673333333333333	15\\
0.676666666666667	15\\
0.68	15\\
0.683333333333333	15\\
0.686666666666667	15\\
0.69	15\\
0.693333333333333	15\\
0.696666666666667	15\\
0.7	15\\
0.703333333333333	15\\
0.706666666666667	15\\
0.71	15\\
0.713333333333333	15\\
0.716666666666667	15\\
0.72	15\\
0.723333333333333	15\\
0.726666666666667	15\\
0.73	15\\
0.733333333333333	15\\
0.736666666666667	15\\
0.74	15\\
0.743333333333333	15\\
0.746666666666667	15\\
0.75	15\\
0.753333333333333	15\\
0.756666666666667	15\\
0.76	15\\
0.763333333333333	15\\
0.766666666666667	15\\
0.77	15\\
0.773333333333333	15\\
0.776666666666667	15\\
0.78	15\\
0.783333333333333	15\\
0.786666666666667	15\\
0.79	15\\
0.793333333333333	15\\
0.796666666666667	15\\
0.8	15\\
0.803333333333333	15\\
0.806666666666667	15\\
0.81	15\\
0.813333333333333	15\\
0.816666666666667	15\\
0.82	15\\
0.823333333333333	15\\
0.826666666666667	15\\
0.83	15\\
0.833333333333333	15\\
0.836666666666667	15\\
0.84	15\\
0.843333333333333	15\\
0.846666666666667	15\\
0.85	15\\
0.853333333333333	15\\
0.856666666666667	15\\
0.86	15\\
0.863333333333333	15\\
0.866666666666667	15\\
0.87	15\\
0.873333333333333	15\\
0.876666666666667	15\\
0.88	15\\
0.883333333333333	15\\
0.886666666666667	15\\
0.89	15\\
0.893333333333333	15\\
0.896666666666667	15\\
0.9	15\\
0.903333333333333	15\\
0.906666666666667	15\\
0.91	15\\
0.913333333333333	15\\
0.916666666666667	15\\
0.92	15\\
0.923333333333333	15\\
0.926666666666667	15\\
0.93	15\\
0.933333333333333	15\\
0.936666666666667	15\\
0.94	15\\
0.943333333333333	15\\
0.946666666666667	15\\
0.95	15\\
0.953333333333333	15\\
0.956666666666667	15\\
0.96	15\\
0.963333333333333	15\\
0.966666666666667	15\\
0.97	15\\
0.973333333333333	15\\
0.976666666666667	15\\
0.98	15\\
0.983333333333333	15\\
0.986666666666667	15\\
0.99	15\\
0.993333333333333	15\\
0.996666666666667	15\\
1	15\\
};

\end{axis}
\end{tikzpicture}%
\caption{Contingency price of the options market under different values of  information state.} 
\label{fig:contingentprice}
\end{minipage}
\begin{minipage}[b]{0.01\linewidth}
\hfill
\end{minipage}
\begin{minipage}[b]{0.32\linewidth}
\centering
\input{figure5_5}
\caption{Options price and strike price as competitive equilibria of  the options market.} 
\label{fig:contingentprice_options}
\end{minipage}
\begin{minipage}[b]{0.01\linewidth}
\hfill
\end{minipage}
\begin{minipage}[b]{0.32\linewidth}
\centering
\input{figure7_7}
\caption{Traded options at the competitive equilibrium under different equilibrium strike price in the options market. }
\label{fig:x_options}
\end{minipage}
\end{figure*}

\begin{figure*}[bt]%
\begin{minipage}[b]{0.32\linewidth}
\centering
\input{figure6_6}
\caption{Energy procurement in the day-ahead market under different strike price in the options market.}
\label{fig:q_options}
\end{minipage}
\begin{minipage}[b]{0.01\linewidth}
\hfill
\end{minipage}
\begin{minipage}[b]{0.32\linewidth}
\centering
%
%
\definecolor{mycolor1}{rgb}{0.60000,0.20000,0.00000}
\begin{tikzpicture}

\begin{axis}[%
width=1.8in,
height=1in,
at={(1.122in,1.067in)},
scale only axis,
xmin=0,
xmax=1,
xlabel style={font=\color{white!15!black}},
xlabel={Information State},
ymin=0,
ymax=0.6,
ylabel style={font=\color{white!15!black}},
ylabel={$y_s$ (MWh)},
axis background/.style={fill=white},
legend style={legend cell align=left, align=left, draw=white!15!black}
]
\addplot [color=mycolor1, line width=1pt, forget plot]
  table[row sep=crcr]{%
0	0.532260759111985\\
0.00333333333333333	0.532260759111985\\
0.00666666666666667	0.532260759111985\\
0.01	0.532260759111985\\
0.0133333333333333	0.532260759111985\\
0.0166666666666667	0.532260759111985\\
0.02	0.532260759111985\\
0.0233333333333333	0.532260759111985\\
0.0266666666666667	0.532260759111985\\
0.03	0.532260759111985\\
0.0333333333333333	0.532260759111985\\
0.0366666666666667	0.532260759111985\\
0.04	0.532260759111985\\
0.0433333333333333	0.532260759111985\\
0.0466666666666667	0.532260759111985\\
0.05	0.532260759111985\\
0.0533333333333333	0.532260759111985\\
0.0566666666666667	0.532260759111985\\
0.06	0.532260759111985\\
0.0633333333333333	0.532260759111985\\
0.0666666666666667	0.532260759111985\\
0.07	0.532260759111985\\
0.0733333333333333	0.532260759111985\\
0.0766666666666667	0.532260759111985\\
0.08	0.532260759111985\\
0.0833333333333333	0.532260759111985\\
0.0866666666666667	0.532260759111985\\
0.09	0.532260759111985\\
0.0933333333333333	0.532260759111985\\
0.0966666666666667	0.532260759111985\\
0.1	0.532260759111985\\
0.103333333333333	0.532260759111985\\
0.106666666666667	0.532260759111985\\
0.11	0.532260759111985\\
0.113333333333333	0.532260759111985\\
0.116666666666667	0.532260759111985\\
0.12	0.532260759111985\\
0.123333333333333	0.532260759111985\\
0.126666666666667	0.532260759111985\\
0.13	0.532260759111985\\
0.133333333333333	0.532260759111985\\
0.136666666666667	0.532260759111985\\
0.14	0.532260759111985\\
0.143333333333333	0.532260759111985\\
0.146666666666667	0.532260759111985\\
0.15	0.532260759111985\\
0.153333333333333	0.532260759111985\\
0.156666666666667	0.532260759111985\\
0.16	0.532260759111985\\
0.163333333333333	0.532260759111985\\
0.166666666666667	0.532260759111985\\
0.17	0.532260759111985\\
0.173333333333333	0.532260759111985\\
0.176666666666667	0.532260759111985\\
0.18	0.532260759111985\\
0.183333333333333	0.532260759111985\\
0.186666666666667	0.532260759111985\\
0.19	0.532260759111985\\
0.193333333333333	0.532260759111985\\
0.196666666666667	0.532260759111985\\
0.2	0.532260759111985\\
0.203333333333333	0.532260759111985\\
0.206666666666667	0.532260759111985\\
0.21	0.532260759111985\\
0.213333333333333	0.532260759111985\\
0.216666666666667	0.532260759111985\\
0.22	0.532260759111985\\
0.223333333333333	0.532260759111985\\
0.226666666666667	0.532260759111985\\
0.23	0.532260759111985\\
0.233333333333333	0.532260759111985\\
0.236666666666667	0.532260759111985\\
0.24	0.532260759111985\\
0.243333333333333	0.532260759111985\\
0.246666666666667	0.532260759111985\\
0.25	0.532260759111985\\
0.253333333333333	0.532260759111985\\
0.256666666666667	0.532260759111985\\
0.26	0.532260759111985\\
0.263333333333333	0.532260759111985\\
0.266666666666667	0.532260759111985\\
0.27	0.532260759111985\\
0.273333333333333	0.532260759111985\\
0.276666666666667	0.532260759111985\\
0.28	0.532260759111985\\
0.283333333333333	0.532260759111985\\
0.286666666666667	0.532260759111985\\
0.29	0.532260759111985\\
0.293333333333333	0.532260759111985\\
0.296666666666667	0.532260759111985\\
0.3	0.532260759111985\\
0.303333333333333	0.532260759111985\\
0.306666666666667	0.532260759111985\\
0.31	0.532260759111985\\
0.313333333333333	0.532260759111985\\
0.316666666666667	0.532260759111985\\
0.32	0.532260759111985\\
0.323333333333333	0.532260759111985\\
0.326666666666667	0.532260759111985\\
0.33	0.532260759111985\\
0.333333333333333	0.532260759111985\\
0.336666666666667	0.532260759111985\\
0.34	0.532260759111985\\
0.343333333333333	0.532260759111985\\
0.346666666666667	0.532260759111985\\
0.35	0.532260759111985\\
0.353333333333333	0.532260759111985\\
0.356666666666667	0.532260759111985\\
0.36	0.532260759111985\\
0.363333333333333	0.532260759111985\\
0.366666666666667	0.532260759111985\\
0.37	0.529133474341146\\
0.373333333333333	0.521809255591146\\
0.376666666666667	0.514485036841146\\
0.38	0.507160818091146\\
0.383333333333333	0.501301443091146\\
0.386666666666667	0.493977224341146\\
0.39	0.486653005591146\\
0.393333333333333	0.479328786841146\\
0.396666666666667	0.473469411841146\\
0.4	0.466145193091146\\
0.403333333333333	0.458820974341146\\
0.406666666666667	0.451496755591146\\
0.41	0.445637380591146\\
0.413333333333333	0.438313161841146\\
0.416666666666667	0.430988943091146\\
0.42	0.423664724341146\\
0.423333333333333	0.417805349341146\\
0.426666666666667	0.410481130591146\\
0.43	0.403156911841146\\
0.433333333333333	0.395832693091146\\
0.436666666666667	0.389973318091146\\
0.44	0.382649099341146\\
0.443333333333333	0.375324880591146\\
0.446666666666667	0.368000661841146\\
0.45	0.362141286841146\\
0.453333333333333	0.354817068091146\\
0.456666666666667	0.347492849341146\\
0.46	0.340168630591146\\
0.463333333333333	0.334309255591146\\
0.466666666666667	0.326985036841146\\
0.47	0.319660818091146\\
0.473333333333333	0.312336599341146\\
0.476666666666667	0.306477224341146\\
0.48	0.299153005591146\\
0.483333333333333	0.291828786841146\\
0.486666666666667	0.284504568091146\\
0.49	0.277180349341146\\
0.493333333333333	0.271320974341146\\
0.496666666666667	0.263996755591146\\
0.5	0.256672536841146\\
0.503333333333333	0.249348318091146\\
0.506666666666667	0.243488943091146\\
0.51	0.236164724341146\\
0.513333333333333	0.228840505591146\\
0.516666666666667	0.221516286841146\\
0.52	0.215656911841146\\
0.523333333333333	0.208332693091146\\
0.526666666666667	0.201008474341146\\
0.53	0.193684255591146\\
0.533333333333333	0.186360036841146\\
0.536666666666667	0.180500661841146\\
0.54	0.173176443091146\\
0.543333333333333	0.165852224341146\\
0.546666666666667	0.158528005591146\\
0.55	0.152668630591146\\
0.553333333333333	0.145344411841146\\
0.556666666666667	0.138020193091146\\
0.56	0.130695974341146\\
0.563333333333333	0.124836599341146\\
0.566666666666667	0.117512380591146\\
0.57	0.110188161841146\\
0.573333333333333	0.102863943091146\\
0.576666666666667	0.0955397243411464\\
0.58	0.0896803493411464\\
0.583333333333333	0.0823561305911464\\
0.586666666666667	0.0750319118411464\\
0.59	0.0677076930911464\\
0.593333333333333	0.0603834743411464\\
0.596666666666667	0.0545240993411464\\
0.6	0.0471998805911464\\
0.603333333333333	0.0398756618411464\\
0.606666666666667	0.0325514430911464\\
0.61	0.0266920680911464\\
0.613333333333333	0.0193678493411464\\
0.616666666666667	0.0120436305911464\\
0.62	0.00471941184114644\\
0.623333333333333	0\\
0.626666666666667	0\\
0.63	0\\
0.633333333333333	0\\
0.636666666666667	0\\
0.64	0\\
0.643333333333333	0\\
0.646666666666667	0\\
0.65	0\\
0.653333333333333	0\\
0.656666666666667	0\\
0.66	0\\
0.663333333333333	0\\
0.666666666666667	0\\
0.67	0\\
0.673333333333333	0\\
0.676666666666667	0\\
0.68	0\\
0.683333333333333	0\\
0.686666666666667	0\\
0.69	0\\
0.693333333333333	0\\
0.696666666666667	0\\
0.7	0\\
0.703333333333333	0\\
0.706666666666667	0\\
0.71	0\\
0.713333333333333	0\\
0.716666666666667	0\\
0.72	0\\
0.723333333333333	0\\
0.726666666666667	0\\
0.73	0\\
0.733333333333333	0\\
0.736666666666667	0\\
0.74	0\\
0.743333333333333	0\\
0.746666666666667	0\\
0.75	0\\
0.753333333333333	0\\
0.756666666666667	0\\
0.76	0\\
0.763333333333333	0\\
0.766666666666667	0\\
0.77	0\\
0.773333333333333	0\\
0.776666666666667	0\\
0.78	0\\
0.783333333333333	0\\
0.786666666666667	0\\
0.79	0\\
0.793333333333333	0\\
0.796666666666667	0\\
0.8	0\\
0.803333333333333	0\\
0.806666666666667	0\\
0.81	0\\
0.813333333333333	0\\
0.816666666666667	0\\
0.82	0\\
0.823333333333333	0\\
0.826666666666667	0\\
0.83	0\\
0.833333333333333	0\\
0.836666666666667	0\\
0.84	0\\
0.843333333333333	0\\
0.846666666666667	0\\
0.85	0\\
0.853333333333333	0\\
0.856666666666667	0\\
0.86	0\\
0.863333333333333	0\\
0.866666666666667	0\\
0.87	0\\
0.873333333333333	0\\
0.876666666666667	0\\
0.88	0\\
0.883333333333333	0\\
0.886666666666667	0\\
0.89	0\\
0.893333333333333	0\\
0.896666666666667	0\\
0.9	0\\
0.903333333333333	0\\
0.906666666666667	0\\
0.91	0\\
0.913333333333333	0\\
0.916666666666667	0\\
0.92	0\\
0.923333333333333	0\\
0.926666666666667	0\\
0.93	0\\
0.933333333333333	0\\
0.936666666666667	0\\
0.94	0\\
0.943333333333333	0\\
0.946666666666667	0\\
0.95	0\\
0.953333333333333	0\\
0.956666666666667	0\\
0.96	0\\
0.963333333333333	0\\
0.966666666666667	0\\
0.97	0\\
0.973333333333333	0\\
0.976666666666667	0\\
0.98	0\\
0.983333333333333	0\\
0.986666666666667	0\\
0.99	0\\
0.993333333333333	0\\
0.996666666666667	0\\
1	0\\
};

\end{axis}
\end{tikzpicture}%
\caption{Load reduction under different information state when strike price minimizes the system total cost.}
\label{fig:reduction_options}
\end{minipage}
\begin{minipage}[b]{0.01\linewidth}
\hfill
\end{minipage}
\begin{minipage}[b]{0.32\linewidth}
\centering
\input{figure9_9}
\caption{System cost of the optimal scheduling without demand response, spot market, and options market.}
\label{fig:systemcostcompare}
\end{minipage}
\end{figure*}

%

\section{Case Studies} 
\label{simulation_sec}
This section illustrates the proposed options market and validates the results through numerical simulation. We consider a particular time interval where the LSE  needs to deliver electricity to a total load of $l=3 \text{MW}$. We emphasize that this is without any loss of generality because larger or smaller load can be captured by scaling the size of the options and load reduction. 
 
Define the information state $s$ as a real number between  $0$ and $1$.  i.e.,  $s \in [0,1]$.  We assume that $s$ is unknown at time $t_0$ (e.g., one day ahead), but it known at time $t_1$ (e.g., one hour ahead). Here we associate $s$ with the average hourly wind energy level on the next day\footnote{Since the realization of wind is a random variable, $s$ is monotonically associated with the expectation of $w$.}. The probability distribution of $s$ and $w$ can be derived based on historic data. In particular, we collect the 5-minute wind generation data between Nov 2019 - May 2020 in California from CAISO \cite{winddata}. The data is scaled based on the size of the case study, where $s=0$  represents zero wind generation, and  $s=1$  represents the maximum wind generation over this period. The histogram of $s$ is shown in Figure \ref{PDF_s}. The empirical distribution of $s$ can be constructed accordingly. To simplify the computation, we approximate the empirical distribution with a polynomial function  to obtain closed-form expression of $\alpha(s)$.

To obtain the conditional probability distribution $p_s(w)$, we represent the wind generation of each time instance as $w=\bar{w}(1+\sigma)$, where $\bar{w}$ denotes the average wind generation of the hour\footnote{Based on our definition, $\bar{w}$ is a monotone function of $s$. Here we set $\bar{w}=2s+0.5$.}, and $\sigma$ is the percentage of deviation from $\bar{w}$. Assume that $\sigma$ is i.i.d. for each data sample within the dataset \cite{winddata} of 5 minute resolution. The empirical distribution of $\sigma$ can be constructed from the dataset. We observe that the empirical distribution of $\sigma$ is symmetric with respect $0$ and displays exponential decay. Therefore,  we approximate it using a symmetric piece-wise exponential function. The empirical distribution (bar plot) and its analytic approximation are shown in Figure~\ref{CDF_s}. Clearly, the approximation is rather accurate.

We randomly select several day-ahead and real-time prices in Berkeley, California, Jan, 2020 \cite{berkeleyLMP}, and derive that the average day-ahead and real-time LMP are $\$26.76/\text{MWh}$ and $\$29.86/\text{MWh}$, respectively. Suppose the conditional expectation of the real-time price is a linear function of $s$ in the form $\bar{\pi}_s^{rt}=31.71- 3.71s$\footnote{We set the parameter of this linear model so that (i) the average real-time price equals the real data, i.e., $\$29.86/\text{MWh}$, and (ii) wind energy affects real-time price by the maximum of $10\%$.}. For simplicity, consider a quadratic distutility function $\phi(y)=15y+15y^2$. In this example, both Assumptions \ref{as:w-pirt} and \ref{disutilityfunction} are satisfied. 

We consider four cases: (a) optimal scheduling without demand response, (b) spot market with contingent pricing,  (c) options market and (d) redesign of options market. We compute  the optimal decisions in each case and derive the competitive equilibrium for these market setups. The system costs in these cases are compared. 

\vspace{0.1cm}
\noindent 1. \emph{Optimal Scheduling without Demand Response:} The LSE buys energy from the day-ahead and real-time markets. Using Proposition \ref{thm:no-dr-social planner}, the optimal day-ahead purchase is $q_{ndr}^e= 1.35 \ \text{MW}$. The optimal system cost is $J_{ndr}^e(q_{ndr}^e) = \$ 55.469$.
%

\vspace{0.1cm}
\noindent 2. \emph{Optimal Scheduling with Demand Response:}  With demand response, the optimal day-ahead purchase is $q_{ndr}^e= 1.19 \ \text{MW}$ and the optimal system cost is $J_{ndr}^e(q_{ndr}^e) = \$ 52.127$.

\vspace{0.1cm}
\noindent 3. \emph{Spot Market with Contingent Price:}  
Solving equation (\ref{optimalconditionofspot}), the optimal day-ahead purchase is $q^*=1.19\ \text{MW}$. The optimal load reduction and contingent price given the information state $s$ are shown in Figures \ref{fig:loadreductionsopt} and \ref{fig:contingentprice} respectively. Figure \ref{fig:loadreductionsopt} reveals that when $s$ is larger, the LSE expects more wind at the delivery time $T$ and therefore  calls on less load reduction at the intermediate time $t_1$. As a result, the competitive equilibrium price is lower when $s$ is larger, as shown in Figure \ref{fig:contingentprice}.   

\vspace{0.1cm}
\noindent 4. \emph{Options Market:} 
The competitive equilibrium of the options market can be derived based on Theorem \ref{thm:ce-option_oldversion} and Theorem~ \ref{thm:ce-option}, respectively. The setup of these two options markets are distinct. Here we solve the competitive equilibria under both market schemes and compare their performances in terms of the system cost at the competitive equilibria. 

We first compute the competitive equilibrium of the options market presented in Section \ref{options_market_original}. Based on Theorem \ref{thm:ce-option_oldversion}, the competitive equilibrium constitutes four decision variables $(\pi^{*o}, \pi^{*sp}, \tilde{x}^*, \tilde{q}^*)$ that are determined by three equations. This indicates that there  is an extra degree of freedom, which induces multiple competitive equilibria. Figure \ref{fig:q_options} shows competitive equilibrium prices of the options market. Clearly, there is a continuum of competitive equilibrium prices and the options price is a decreasing function of the strike price. This is intuitive since the objective functions of the LES and the aggregator are jointly determined by $\pi^o$ and $\pi^{sp}$. Therefore, for any pair of $(\pi^o, \pi^{sp})$ at the competitive equilibrium, the combination of a higher $\pi^o$ and lower $\pi^{sp}$ is equally acceptable as another competitive equilibrium price. Figure \ref{fig:x_options} shows that as the strike price increases, the traded options increases due to the decrease of options price. Figure \ref{fig:q_options} shows that the day-ahead purchase $q$ at the competitive equilibrium first decreases $(\pi^{*sp}\leq \$3.4/\text{WMh})$ and then increases $(\$3.4/\text{WMh}< \pi^{*sp}\leq \$31.3/\text{WMh})$. The first regime $(\pi^{*sp}\leq \$3.4/\text{WMh})$ is accompanied with a sharp increase in traded options as shown in Figure \ref{fig:x_options}. Figure \ref{fig:reduction_options} shows the exercised load reduction in real time under different information state when the strike price is at $\$22.6/\text{WMh}$. Clearly, $y_s$ is a decreasing function of  $s$, which indicates that less load reduction is called when more wind energy is available for free.

Figure \ref{fig:systemcostcompare} compares the system costs at the competitive equilibria of the original (Section \ref{options_market_original}) and the redesigned  options market  (Section \ref{options_market_redesign}) under different equilibrium strike prices. The equilibrium of the redesigned market is derived based on  (\ref{eq:ce-option-ys}).  It is clear that the system cost at the competitive equilibrium for the redesigned market is  consistently lower than than that of the original one, while the difference of these two markets is more significant when the strike price is smaller. At the optimal strike price that minimizes the system cost, both the original options market and the redesigned market attained a cost of $52.154$, which is close to the system cost for the spot market,i.e., $52.127$. This indicates that by carefully choosing the strike price, the options market can almost achieve the same system cost as the efficiency spot market. However, obtaining the optimal options market price may require information not possessed by the system operator.

\section{Conclusion}
\label{sec:conclusion}
We have studied a novel market model for trading demand response using options. We have shown that demand response can be used as an intermediate recourse between the real-time market and the day-ahead market. Under some conditions, this options market admits a competitive equilibrium. We studied the efficiency of this equilibrium, and obtained the optimal strike price that yields the minimum system cost at the competitive equilibrium. In future work we plan to address  option markets with multiple intermediate stages and also the case where the LSE can exercise market power. 

\bibliographystyle{IEEEtran}
\bibliography{Options-References}
 
\section{Proofs}
\subsection{Proof of 
Proposition \ref{thm:no-dr-social planner}}
The objective function $J_{ndr}^e$ is,
\begin{equation}
\label{eq:social planner-J-ndr-2}
J^{e}_{ndr}(q) = \pi^{da} q + \int^{1}_{s=0} \int^{(l-q)}_{w=0} \overline{\pi}^{rt}_{s} (l-q-w) p_{s}(w) dw ~\alpha(s) ds
\end{equation}
By taking partial derivative we get,
\begin{align*}
\frac{d J^{e}_{ndr}}{d q} &= \pi^{da} -  \int^{1}_{s=0} \overline{\pi}^{rt}_{s} P_{s}(l-q) \alpha(s) ds \\
\frac{d^{2}J^{e}_{ndr}}{d q^{2}} &= \int^{1}_{s=0} \overline{\pi}^{rt}_{s} p_{s}(l-q) \alpha(s) ds
\end{align*}
Since the second derivative is non-negative, $J^{e}_{ndr}$ is convex. The solution is obtained by equating the first derivative to zero.

\subsection{Proof of Proposition \ref{thm:dr-social planner}} 
\begin{proof}
The cost of the second stage is,
\begin{equation*}
{J}^{e}_{s}(y_{s}) = \phi(y_{s}) + \int^{(l-q-y_{s})}_{w=0} \overline{\pi}^{rt}_{s} (l-q-w-y_{s}) p_{s}(w) dw.
\end{equation*} 
Taking partial derivatives, we have,
\begin{align*} 
\frac{\partial J^{e}_{s}}{\partial y_{s}} &= \phi'(y_{s}) - \overline{\pi}^{rt}_{s} P_{s}(l-q-y_{s}) \\
\frac{\partial^{2} J^{e}_{s}}{\partial y_{s}^{2}} &= \overline{\pi}^{rt}_{s} p_{s}(l-q-y_{s})
\end{align*} 
Since the second order derivative is non-negative, $J^{e}_{s}(y_{s})$ is convex. Using first derivative, if $\phi(y_s)' \geq  \overline{\pi}^{rt}_{1} P_{s}(l-q)$, then the optimal load curtailment is
$y^{e}_{s} = 0$. If
$\phi(y_s)' < \overline{\pi}^{rt}_{s} P_{s}(l-q)$, $y^e_s$ is given by equating the first order condition to zero,
\begin{equation}
\label{firstderivativeeqaul0}
\frac{\partial J^{e}_{s}}{\partial y_{s}} = \phi(y_s)' - \overline{\pi}^{rt}_{s} P_{s}(l-q-y_{s})\big \vert_{y_{s} = y^e_{s}}  = 0.
\end{equation}
This leads to (\ref{eq:social planner-ys-1}). Next, we compute the right derivative of the first stage cost $J^e(q)$,
\begin{align}
\label{directtiond}
\dfrac{d^+J^e}{dq}&=\pi^{da} + \int_{0}^1  \frac{\partial^+ y^e_{s}}{\partial q } \phi'(y_s)\alpha(s)ds. \nonumber \\
&-  \int^{1}_{0} (1+\frac{\partial^+ y^e_{s}}{\partial q })  \overline{\pi}^{rt}_{s} P_{s}(l-q-y^e_{s}) \alpha(s) ds
\end{align}
Here $\frac{\partial^+ y^e_{s}}{\partial q }$ is the right partial derivative of $y_{s}^e$ with respect to $q$. Note that when $\phi'(y_s)\geq  \overline{\pi}^{rt}_{s} P_{s}(l-q)$, $\frac{\partial^+ y^e_{s}}{\partial q } = 0$. When $\phi'(y_s) <  \overline{\pi}^{rt}_{s} P_{s}(l-q)$, we have (\ref{firstderivativeeqaul0}). In either case, (\ref{directtiond}) is equivalent to,
\begin{equation}
\label{dervativeof1st}
\dfrac{d^+J^e}{dq}=\pi^{da} -  \int^{1}_{0} \overline{\pi}^{rt}_{s} P_{s}(l-q-y^e_{s}) \alpha(s) ds. 
\end{equation}
Similarly, we can derive the left derivative of $J^e(q)$. It equals the right derivative. Therefore, $J^e(q)$ is differentiable with respect to $q$. The second-order right derivative of $J^e(q)$ is as follows,
\begin{equation*}
\frac{d^{2} J^{e}}{d q^{2}} =  \int^{1}_{0} \overline{\pi}^{rt}_{s} p_{s}(l-q-y^e_{s}) \left(1 + \frac{\partial^{+} y^e_{s}}{\partial q } \right)\alpha(s)ds.
\end{equation*}
To prove Proposition \ref{thm:dr-social planner}, it suffices to show that $J^e(q)$ is convex with respect to $q$. If so, the optimal decision is obtained by equating the first derivative to $0$, which leads to (\ref{eq:social planner-qe-1}). To this end, we show that $\frac{d^{2} J^{e}}{d q^{2}}\geq 0$.  
Note that when $\phi'(y_s) \geq \overline{\pi}^{rt}_{s} P_{s}(l-q)$, $\frac{\partial^+ y^e_s}{\partial q} = 0$. It trivially holds. When $\phi'(y_s) < \overline{\pi}^{rt}_{s} P_{s}(l-q)$, we take partial derivative with respect to $q$ on both sides of (\ref{firstderivativeeqaul0}). This indicates that,
 \[ \overline{\pi}^{rt}_{s} p_{s}(l-q-y^e_{s})  \left(1 + \frac{\partial^+ y^e_{s}}{\partial q } \right)=0\]   
Therefore, $\frac{d^{2} J^{e}}{d q^{2}}=0$. This completes the proof. 
\end{proof}

\subsection{Proof of Theorem \ref{thm:ce-market} } 
\begin{proof}
It suffices to show there exists a solution to (\ref{eq:defineCEforspot}), and any solution to (\ref{eq:defineCEforspot}) satisfies (\ref{optimalconditionofspot}). According to Proposition \ref{thm:dr-social planner}, it suffices to show that (\ref{eq:defineCEforspot}) is equivalent to (\ref{eq:social planner-ys-1}) and (\ref{eq:social planner-qe-1}). Given any $q$, the optimal solution to (\ref{eq:defineCEforspotb}) is given by
\begin{align}
\label{optimslsolutionyLSE}
y_{s}^{LSE}=
\begin{cases}
l-q-P_s^{-1}(\pi^{in}_{s}/\bar{\pi}_s^{rt}),  & \text{if } \pi^{in}_{s} < \overline{\pi}^{rt}_{s} P_{s}(l-q), \\
0,  & \text{if } \pi^{in}_{s} \geq  \overline{\pi}^{rt}_{s} P_{s}(l-q).
\end{cases}
\end{align}
For any $q$, the optimal solution to (\ref{eq:defineCEforspotc}) is given by,
\begin{align}
\label{optimslsolutionyagg}
y_s^{agg}=
\begin{cases}
0, & \text{if}~  \phi'(0)>\pi_s^{in}, \\
(\phi')^{-1}(\pi_s^{in}), & \text{if}~  \phi'(0)\leq\pi_s^{in}
\end{cases} 
\end{align}  
Proof for (\ref{optimslsolutionyLSE}) and (\ref{optimslsolutionyagg}) is similar to that of Proposition \ref{thm:dr-social planner} and hence we skip the details. Since (\ref{optimslsolutionyLSE}) is continuous and decreasing, and (\ref{optimslsolutionyagg}) is continuous and increasing, there is an intersection. Thus the competitive equilibrium exists.	
It can be verified that at the intersection, the equilibrium satisfies (\ref{eq:social planner-ys-1}).
In addition, based on (\ref{optimslsolutionyLSE}), the derivative of the first stage cost of the LSE is given by,
\begin{equation}
\label{derivativeoffirststage}
\dfrac{d J^{LSE}}{dq}=\pi^{da} -  \int^{1}_{0}   \overline{\pi}^{rt}_{s} P_{s}(l-q-y^{LSE}_{s}) \alpha(s) ds. 
\end{equation}
At the optimal decision $q^{*LSE}$, we have,
\begin{equation}
\label{deducedoptimalq}
\pi^{da} = \mathbb{E}_{s}[ \overline{\pi}^{rt}_{s} P_{s}(l-q^{*LSE}-y^{*}_{s})]  = 0.
\end{equation}
Therefore, (\ref{eq:defineCEforspot}) is equivalent to (\ref{eq:social planner-ys-1}) and (\ref{eq:social planner-qe-1}). This completes the proof.
\end{proof}  

\subsection{Proof of Theorem \ref{thm:ce-option_oldversion}}
    
We first prove a series of lemmas before giving the proof of Theorem \ref{thm:ce-option_oldversion}. 
\begin{lemma}
The function $\tilde{J}^{LSE}_{s}(\cdot)$ is convex for all $s \in S$. The unique minimizer $\tilde{y}^{LSE}_{s}$ is given by,
\begin{align}
\label{secondstageoptimaldecision}
\tilde{y}^{LSE}_{s} &= \left\lbrace \begin{array}{ccc}
0, \quad \quad  \quad \quad \quad \quad \quad \quad \text{if}~ P_{s}(l-q) < \pi^{sp}/\pi^{rt}_{s}  \\
x, \quad \quad \quad \quad \quad \quad  \text{if}~P_{s}(l-q-x) > \pi^{sp}/\pi^{rt}_{s} \\
l - q- P^{-1}_{s}(\dfrac{\pi^{sp}}{\pi^{rt}_{s}}),  \quad \quad \quad \quad \quad \quad \textnormal{otherwise}
\end{array} \right.
\end{align}
\end{lemma}
\begin{proof}
The second-stage cost for the LSE is,
\[ \tilde{J}^{LSE}_{s}(y_s) = \pi^{sp} y_s +   \overline{\pi}^{rt}_{s} \int^{l-q-y_s}_{0}(l - w - q - y_s) p_{s}(w) dw,  \]
The first-order and second-order derivative are as follows,
\begin{align}
\begin{cases}
\frac{\partial \tilde{J}^{LSE}_{s}}{\partial y_s}  = \pi^{sp} - \overline{\pi}^{rt}_{s} P_{s}(l-q-y_s), \nonumber \\
\frac{\partial^{2} \tilde{J}^{LSE}_{s}}{\partial y_s^{2}} =  \overline{\pi}^{rt}_{s} p_{s}(l-q-y_s).
\end{cases}
\end{align}
Since the second derivative is positive, $\tilde{J}^{LSE}_{s}(y)$ is strictly convex. Then using first derivative and the strict convexity property the expression for the unique minimizer $\tilde{y}^{LSE}_{s}$ follows.
\end{proof}

\begin{lemma}
\label{lem:ce-options-LSE-foc}
The function $\tilde{J}^{LSE}(q, x)$ is jointly convex in $q$ and $x$. The minimizers $(\tilde{q}^{LSE}, x^{LSE})$ are given by,
\begin{align}
\pi^{da} & -\mathbb{E}_{s}[\overline{\pi}^{rt}_{s}P_{s}(l-\tilde{q}^{LSE}-\tilde{y}^{LSE}_{s})] = 0 \nonumber \\
\pi^{o} &+ \pi^{sp} \mathbb{E}_{s}[\mathbb{I}\{y_{s} = x^{LSE}\}] \nonumber  \\
& - \mathbb{E}_{s}[\overline{\pi}^{rt}_{s} P_{s}(l-\tilde{q}^{LSE}-\tilde{y}^{LSE}_{s})  \mathbb{I}\{y_{s} = x^{LSE}\}] = 0
\label{eq:ce-options-LSE-foc}
\end{align}
\end{lemma}
\begin{proof}
Let $s_{1}, s_{2} \in S$ be such that $y_{s} = x$ for $0 \leq s \leq s_{1}$ and $y_{s} = 0$ for $s_{2} \leq s \leq 1$. Note that 
$s_{1}$ and $s_{2}$ depends on $q$ and $x$ from the first stage. 
Then, $\tilde{J}^{LSE}(q, x)$ (c.f. \eqref{eq:J-LSE-option-1}) can be written as,
\begin{align*} 
& \tilde{J}^{LSE}(q,x) =  (\pi^o x + \pi^{da} q)  \\
& + \int^{s_{1}}_{s=0} \pi^{sp}x\alpha(s) ds \\
& + \int^{s_{1}}_{s=0} \overline{\pi}^{rt}_s \int^{(l-q-x)}_{w=0}(l-q-x-w) p_{s}(w)dw\alpha(s) ds  \\
& + \int^{s_{2}}_{s=s_{1}} \pi^{sp} y_{s} \alpha(s) ds \\ 
& + \int^{s_{2}}_{s=s_{1}}\overline{\pi}^{rt}_s\int^{(l-q-y_{s})}_{w=0} (l-q-y_{s}-w)p_{s}(w)dw\alpha(s) ds   \\
& + \int^{1}_{s=s_{2}}\left(\overline{\pi}^{rt}_s \int^{(l-q)}_{w=0} (l-q-w) p_{s}(w) dw \right) \alpha(s) ds   
\end{align*} 

We give simplified expressions for the partial derivatives of $\tilde{J}^{LSE}$ w.r.t $q$ and $x$ below,
\begin{align*}
\frac{\partial \tilde{J}^{LSE}}{\partial q} &= \pi^{da} - \int^{1}_{s=0} \overline{\pi}^{rt}_{s} P_{s}(l-q-\tilde{y}^{LSE}_{s}) \alpha(s) ds\\
\frac{\partial \tilde{J}^{LSE}}{\partial x} &= \pi^{o} + \pi^{sp} \int^{s_{1}}_{s=0}  \alpha(s) ds \\
&\hspace{1cm} -   \int^{s_{1}}_{s=0}  \overline{\pi}^{rt}_{s} P_{s}(l-q-x) \alpha(s) ds  
\end{align*}
Once again differentiating the above expressions w.r.t $q$ and $x$ we get,
\begin{align*}
\frac{\partial^{2} \tilde{J}^{LSE}}{\partial x \partial q} &= \int^{s_{1}}_{s=0} \overline{\pi}^{rt}_s p_{s}(l-q-x) \alpha(s) ds \\
\frac{\partial^{2} \tilde{J}^{LSE}}{\partial q^{2}} &= \int^{s_{1}}_{s=0} \overline{\pi}^{rt}_s p_{s}(l-q-x) \alpha(s) ds \\
&\hspace{0.5cm} + \int^{1}_{s=s_{2}} \overline{\pi}^{rt}_s p_{s}(l-q) \alpha(s) ds \\
\frac{\partial^{2} \tilde{J}^{LSE}}{\partial x^{2}} &= \int^{s_{1}}_{s=0} \overline{\pi}^{rt}_s p_{s}(l-q-x) \alpha(s) ds
\end{align*}

It follows that the Hessian will be of the form 
\begin{align*}
&\left[\begin{array}{cc}
(a+b) & a \\
a & a
\end{array}  \right], ~\text{where}~ a =  \int^{s_{1}}_{s=0} \overline{\pi}^{rt}_{s} p_{s}(l-q-x) \alpha(s) ds, \\
& b =  \int^{1}_{s=s_{2}} \overline{\pi}^{rt}_{s} p_{s}(l-q) \alpha(s) ds.
\end{align*}
It is easy to show that $a > 0$ and $b > 0$ always. We give a very brief argument here. If $\pi^{sp}$ is large then $s_2 < 1$ trivially and so $b > 0$ trivially. But $\pi^{sp}$ cannot be set large enough that $s_2 =  0$ because it obviates the need for an intermediate DR market. Then there will necessarily exist a $s_1 > 0$ and so $a > 0$ in this case. If  $\pi^{sp}$ is small then $s_1 > 0$ trivially and there will always exist a $s_2 < 1$ such that $\overline{\pi}^{rt}_sP_s(l-q) = \pi^{sp}$. It follows that, $a > 0$ and $b > 0$ in this case as well. Hence, by Silvester's criterion, the Hessian is positive definite. Hence, by convexity the minimizers of LSE cost satisfy \eqref{eq:ce-options-LSE-foc}.
\end{proof}

\begin{lemma}
\label{aggconvex}
The cost function of the aggregator $\tilde{J}^{agg}(x)$ is convex in $x$.
\end{lemma}
\begin{proof}
Define $s_1$ and $s_2$ in the same way as in the proof of Lemma \ref{lem:ce-options-LSE-foc}. Recall that $s_1$ denotes the information state below which the LSE schedules all of the options and $s_2$ denotes the information state above which the LSE does not schedule any of the demand response at all. These variables depend on the decision of the LSE, which are in turn only dependent on the option prices $\pi^o$ and $\pi^{sp}$, which are fixed. As a result, $s_1$ and $s_2$ will not be affected by the aggregator's decision $x$. 

The cost function $\tilde{J}^{agg}(x)$ is as follows:
\begin{align*} 
\tilde{J}^{agg} &(x) = -\pi^ox+ \int_{s=0}^{s_1}(\phi(x)-\pi^{sp} x) \alpha(s) ds  \\
& + \int_{s=s_1}^{s_2}(\phi(\tilde{y}_s^{LSE})-\pi^{sp} \tilde{y}_s^{LSE}) \alpha(s) ds +\int_{s=s_2}^1 \phi(0) \alpha(s)ds.   
\end{align*} 
The first order derivative of $\tilde{J}^{agg} (x)$ is given by,
\begin{equation}
\label{derivativeofagg}
\frac{\partial \tilde{J}^{agg}}{\partial x }  = - \pi^{o} + \int^{s_{2}}_{s=0} \left(\phi'(\tilde{y}_s^{LSE})-\pi^{sp}  \right) \alpha(s) ds.
\end{equation}
The second order derivative is given by,
\begin{equation}
\label{secondderivativeofagg}
\frac{\partial^{2} \tilde{J}^{agg}}{\partial x^{2} }  = \int^{s_{2}}_{s=0} \phi''(x) \alpha(s).
\end{equation}
Clearly, $\tilde{J}^{agg}(x)$ is convex. 
\end{proof}

Next we show the continuity of  $\tilde{q}^{LSE}(\pi^{o}, \pi^{sp})$ and $x^{LSE}(\pi^{o}, \pi^{sp})$ in $(\pi^{o}, \pi^{sp})$ using the implicit function theorem. We omit the details for the proof of continuity of $x^{agg}(\pi^o,\pi_{LSE})$.
 
\begin{lemma}
\label{lem:ce-options-LSE-min}
The minimizers of LSE cost $\tilde{J}^{LSE}(q, x)$ i.e. $(\tilde{q}^{LSE}(\pi^o,\pi^{sp}), \tilde{x}^{LSE}(\pi^o,\pi^{sp}))$ are continuous in its arguments. 
\end{lemma}
\begin{proof}
By Lemma \ref{lem:ce-options-LSE-foc}, the minimizers $\tilde{q}^{LSE}, x^{LSE}$ satisfy conditions \eqref{eq:ce-options-LSE-foc}. Define, 
\begin{align}
f(q,x,\pi^o,\pi^{sp}) & = \left( \begin{array}{c} f_1(q,x,\pi^o,\pi^{sp}) \\ f_2(q,x,\pi^o,\pi^{sp})  \end{array} \right) \nonumber \\
\tn{Where} ~~ f_1(q,x,\pi^o,\pi^{sp}) & = \frac{\partial \tilde{J}^{LSE}}{\partial q} \nonumber \\
& =  \pi^{da} - \mathbb{E}_{s}[\overline{\pi}^{rt}_{s} P_{s}(l-q-y_{s})] \nonumber \\
f_2(q,x,\pi^o,\pi^{sp}) & = \frac{\partial \tilde{J}^{LSE}}{\partial x} \nonumber \\
& = \pi^{o} + \pi^{sp} \mathbb{E}_{s}[\mathbb{I}\{y_{s} = x\}] \nonumber  \\
& - \mathbb{E}_{s}[\overline{\pi}^{rt}_{s} P_{s}(l-q-y_{s}) \mathbb{I}\{y_{s} = x\}]
\end{align}
Then the minimizers $\tilde{q}^{LSE}, x^{LSE}$, satisfy $f = 0$. From Lemma \ref{lem:ce-options-LSE-foc} the partial derivatives of $\frac{\partial \tilde{J}^{LSE}}{\partial q}$ and $\frac{\partial \tilde{J}^{LSE}}{\partial x}$ exist. Hence $f(q,x,\pi^o,\pi^{sp})$ is continuously differentiable w.r.t $q$ and $x$. Also the derivatives of $f_1$ and $f_2$ w.r.t $\pi_o $ and $\pi^{sp}$ exists and is given by,
\begin{align*}
\frac{\partial f}{\partial \pi^{o}} & =  \left( \begin{array}{c} 0 \\ 1 \end{array}\right),~~ 
\frac{\partial f}{\partial \pi^{sp}} & = \left( \begin{array}{c} - \mathbb{E}_{s}[\mathbb{I}\{0 < y_{s} < x\}]  \\ \mathbb{E}_s[\mathbb{I}\{y_s = x\}] \end{array} \right) \nonumber \\
\end{align*}
This implies that $f$ is continuously differentiable w.r.t $q, x, \pi^o$ and $\pi^{sp}$. From Lemma \eqref{lem:ce-options-LSE-foc},
\[ H = \left[ \begin{array}{cc} \frac{\partial f}{\partial q} &  \frac{\partial f}{\partial x} \end{array} \right] > 0 \] 
at points $(q,x,\pi^o,\pi^{sp})$ where $\pi^o > 0, \pi^{sp} \geq \pi^{da} $ and $f = 0$.
Then by implicit function theorem there exists continuous functions $g_1: (\pi^o, \pi^{sp}) \rightarrow q$ and $g_2: (\pi^o, \pi^{sp}) \rightarrow x$ such that the minimizers of LSE cost $\tilde{J}^{LSE}(q, x)$ are given by $\tilde{q}^{LSE} = g_1(\pi^o, \pi^{sp})$ and $x^{LSE} = g_2(\pi^o, \pi^{sp})$. Hence the minimizers $(\tilde{q}^{LSE}(\pi_o,\pi_{LSE}), x^{LSE}(\pi^o,\pi^{sp}))$ are continuous in its arguments.
\end{proof}

Similarly we can show that the minimizer of the aggregator's cost, $x^{agg}(\pi^o, \pi^{sp})$ is a continuous function of its arguments. 

We now show the existence of competitive equilibrium for the options market.  Let 
\[z(\pi^{o}, \pi^{sp}) = x^{LSE}(\pi^{o}, \pi^{sp}) - x^{agg}(\pi^{o}, \pi^{sp}) \]
By Lemma  \ref{lem:ce-options-LSE-min}, $z(\pi^{o}, \pi^{sp})$ is a continuous function of $\pi^{o}$ and $\pi^{sp}$. According to the definition, prices $\pi^{*o}, \pi^{*LSE}$ supports an equilibrium if $z(\pi^{*o}, \pi^{*LSE}) = 0$. So, it is sufficient to show that there exists two pair of prices $(\pi^{o}_{1}, \pi^{sp}_{1})$ and $(\pi^{o}_{2}, \pi^{sp}_{2})$ such that
\[z(\pi^{o}_{1}, \pi^{sp}_{1}) > 0,~  z(\pi^{o}_{2}, \pi^{sp}_{2}) < 0\]
Then the existence of prices $(\pi^{*o}, \pi^{*LSE})$ at which $z(\pi^{*o}, \pi^{*LSE}) = 0$ follows from the continuity of the function $z(\cdot, \cdot)$. 


If $\pi^{o}_{1}$ and $\pi^{sp}_{1}$ are very small, the LSE will prefer to get more load reduction. But the aggregator may not want to offer any load reduction because the reward is not enough to offset the disutility from load curtailment. So, $x^{LSE}(\pi^{o}_{1}, \pi^{sp}_{1}) - x^{agg}(\pi^{o}_{1}, \pi^{sp}_{1})$ will be positive. 

On the other hand, if $\pi^{o}_{2}$ and $\pi^{sp}_{2}$ are very high, the aggregator may offer more load curtailment. But the LSE may prefer very small or no load curtailment at all. It may be better for the LSE to purchase energy from DAM or RTM. So, in this case, $x^{LSE}(\pi^{o}_{2}, \pi^{sp}_{2}) - x^{agg}(\pi^{o}_{2}, \pi^{sp}_{2})$ will be negative.  

Now, the existence of $(\pi^{*o}, \pi^{*LSE})$ follows from the continuity of $z(\cdot, \cdot)$. 

\subsection{Proof of Theorem \ref{thm:ce-option} }  
  
We first prove the following lemma before giving the proof of Theorem \ref{thm:ce-option}. 

\begin{lemma}
\label{lem:ce-options-LSE-foc_redesign}
The function $\tilde{J}^{LSE}(x)$ is convex with respect to $x$. 
\end{lemma}
\begin{proof}
We first plug in the second-stage decision (\ref{secondstageoptimaldecision}) in the net expected cost function (\ref{eq:J-LSE-option-1}). We define two variables, $s_1$ and $s_2$, as follows: $s_1$ and $s_2$ are such that,
\begin{align}
\label{definitionofs1s2}
\tilde{y}^{LSE}_{s} =
\begin{cases}
0, &s_2\leq s\leq 1, \\
x, & 0 \leq s \leq  s_1, \\
l - l'+x- P^{-1}_{s}(\pi^{sp}/\pi^{rt}_{s}), & s_1\leq s\leq s_2.
\end{cases}
\end{align} 

Then, $\tilde{J}^{LSE}(x)$ (c.f. \eqref{eq:J-LSE-option-1}) is given by,
\begin{align*} 
& \tilde{J}^{LSE}(x) =  \pi^o x + \pi^{da}(l-l')+ \int^{s_{1}}_{0} \pi^{sp}x\alpha(s) ds\\
& + \int^{s_{2}}_{s_{1}} \pi^{sp} \left( l - l'+x- P^{-1}_{s}(\dfrac{\pi^{sp}}{\pi^{rt}_{s}}) \right) \alpha(s) ds \\
& + \int^{s_{1}}_{s=0}  \int^{l-l'}_{w=0}\overline{\pi}^{rt}_s(l-l'-w) p_{s}(w)dw\alpha(s) ds  \\ 
& + \int^{s_{2}}_{s=s_{1}}\int^{P^{-1}_{s}(\pi^{sp}/\pi^{rt}_{s})}_{w=0}\overline{\pi}^{rt}_s (P^{-1}_{s}(\dfrac{\pi^{sp}}{\pi^{rt}_{s}})-w)p_{s}(w)dw\alpha(s) ds   \\
& + \int^{1}_{s=s_{2}} \int^{l-l'+x}_{w=0} \overline{\pi}^{rt}_s (l-l'+x-w) p_{s}(w) dw  \alpha(s) ds   
\end{align*} 

Using Leibniz rule, the first order derivative of $\tilde{J}^{LSE}$ with respect to $x$ is given by,
\begin{align}
\label{derivativeofLSE}
\frac{\partial \tilde{J}^{LSE}}{\partial x} = &\pi^{0}-\pi^{da} +\int_{s=0}^{s_2} \pi^{sp} \alpha(s)ds \nonumber \\
 &+\int^{1}_{s=s_2} \overline{\pi}^{rt}_{s} P_{s}(l-l'+x) \alpha(s) ds
\end{align}
Then the second order derivative is given by,
\begin{equation*}
\frac{\partial^{2} \tilde{J}^{LSE}}{\partial x^{2}} = \int^{s_{1}}_{s=0} \overline{\pi}^{rt}_s p_{s}(l-q-x) \alpha(s) ds 
\end{equation*}
Since the second order derivative is positive, $\tilde{J}^{LSE}(x)$ is convex.
\end{proof}

Similar to the proof of Theorem \ref{thm:ce-option} and using the above lemma we can show that there exists an intersection to the supply and demand curve, which is the competitive equilibrium.  
 
Next we show that the competitive equilibrium is the solution to (\ref{eq:ce-option-ys}). Let $(\pi^{*o},  \tilde{x}^{*LSE}, \tilde{y}_s^{LSE}, \tilde{x}^{*agg})$ be a competitive equilibrium.
For notation convenience, we define the following,
\begin{align*}
\begin{cases}
\tilde{J}^{LSE}(x)\triangleq C^{LSE}(x)+\pi^ox, \\
\tilde{J}^{agg}(x)\triangleq C^{agg}(x)-\pi^ox.
\end{cases}
\end{align*}
By definition, the competitive equilibrium satisfies
\begin{subnumcases}{\label{competitive_opt_equi}}
 \tilde{x}^{*LSE}= \arg \min_{0\leq x\leq l'} C^{LSE}(x)+\pi^ox,  \label{competitive_opt_equia}\\
\tilde{x}^{*agg}= \arg \min_{0\leq x \leq l'} C^{agg}(x)-\pi^ox \label{competitive_opt_equib} \\
  \tilde{x}^{*LSE}=\tilde{x}^{*agg}. \label{competitive_opt_equic}
\end{subnumcases}
Let $x^*= \tilde{x}^{*LSE}=\tilde{x}^{*agg}$, then (\ref{competitive_opt_equia}) is same as,
\begin{equation}
\label{optimalityLSE}
C^{LSE}(x^*)+\pi^ox^*\leq C^{LSE}(x)+\pi^ox, \quad \forall 0\leq x\leq l',
\end{equation}
and (\ref{competitive_opt_equib}) is same as,
\begin{equation}
\label{optimalityagg2}
C^{agg}(x^*)-\pi^ox^*\leq C^{agg}(z)-\pi^oz, \quad \forall 0\leq z\leq l'.
\end{equation}
Let $z=x$, then (\ref{optimalityLSE}) plus (\ref{optimalityagg2}) gives,
\begin{equation*}
C^{agg}(x^*)+C^{LSE}(x^*)\leq C^{agg}(x)+C^{LSE}(x), \quad \forall 0\leq x\leq l'.
\end{equation*}
This is equivalent to (\ref{eq:ce-option-ys}), which completes the proof.

\subsection{Proof of Proposition \ref{efficiencyofoptions}}
  According to Theorem \ref{thm:ce-market}, $J^{*cp}$ equals the optimal value of $\min_q J^e(q)$.  Clearly, $\min_q J^e(q)$ is smaller than the optimal value of (\ref{eq:ce-option-ys}). This is because (\ref{eq:social planner-J-dr}) optimizes over $y_s$, while (\ref{eq:ce-option-ys}) fixes $y_s$ according to (\ref{optimaLSEcondstagedec}). Therefore,  $J^{*cp}\leq \tilde{J}^{*cp}(\pi^{sp})$. 

To show that $\tilde{J}^{*cp}(\pi^{sp})\leq J_{ndr}^{*e}$, we simply note that when $x=0$, the value of (\ref{eq:ce-option-ys}) attains $J_{ndr}^{*e}$. Therefore, the optimal value of  (\ref{eq:ce-option-ys}) is smaller than $J_{ndr}^{*e}$. This completes the proof.

\subsection{Proof of Theorem \ref{optimalstrikeprice} } 

Based on Theorem \ref{thm:ce-option}, the optimal strike price is the optimal solution to the following problem,
\begin{equation}
\label{eq:ce-option-ys-proof}
\min_{ \pi^{sp}, x\in [0,l']} \pi^{da} (l'-x)+ \mathbb{E}_{s} [\phi(\tilde{y}_s^{LSE})+\tilde{J}_s^{LSE}(\tilde{y}_s^{LSE})],
\end{equation}
where $\tilde{y}_s^{LSE}$ is defined as (\ref{optimaLSEcondstagedec}). The first order derivative of (\ref{eq:ce-option-ys-proof}) with respect to $\pi^{sp}$ can be computed as,
\begin{align*}
\dfrac{\partial \tilde{J}^e(x,\pi^{sp})}{\partial \pi^{sp}}&=\int_{s_1}^{s_2} \dfrac{1}{\bar{\pi}_s^{rt}} (P_s^{-1})'(\dfrac{\pi^{sp}}{\bar{\pi}_s^{rt}})[\pi^{sp}-\phi'(y_s)]\alpha(s)ds \nonumber \\
&=\int_{s_1}^{s_2} \dfrac{[\pi^{sp}-\phi'(y_s)]}{\bar{\pi}_s^{rt}p_s(l-l'+x-y_s)} \alpha(s)ds,
\end{align*}
where the second equation is based on the fact that 
\begin{equation*}
(P_s^{-1})'(\dfrac{\pi^{sp}}{\bar{\pi}_s^{rt}})=\dfrac{1}{p_s(P_s^{-1}(\pi^{sp}/\bar{\pi}_s^{rt}))}=\dfrac{1}{p_s(l-l'+x-y_s)}.
\end{equation*}
Note that $\dfrac{\partial \tilde{J}^e(x,\pi^{sp})}{\partial \pi^{sp}}\leq 0$ when $\pi^{sp}\leq \phi'(0)$, and $\dfrac{\partial \tilde{J}^e(x,\pi^{sp})}{\partial \pi^{sp}}\geq 0$ when $\pi^{sp}\geq \phi'(l')$. Since $\dfrac{\partial \tilde{J}^e(x,\pi^{sp})}{\partial \pi^{sp}}$ is continuous with respect to $\pi^{sp}$, there exists an optimal strike price $\tilde{\pi}^{*LSE}$, such that $\phi'(0)\leq  \tilde{\pi}^{*LSE} \leq \phi'(l')$  and $\dfrac{\partial \tilde{J}^e(x,\tilde{\pi}^{*LSE})}{\partial \pi^{sp}}=0$. The final expression fpr $\tilde{\pi}^{*LSE}$ follows trivially from the condition $\dfrac{\partial \tilde{J}^e(x,\tilde{\pi}^{*LSE})}{\partial \pi^{sp}}=0$. This completes the proof.

\end{document}